\documentclass[10pt,twocolumn]{article}

\usepackage[a4paper,margin=1.7cm,columnsep=0.62cm]{geometry}

\usepackage[T1]{fontenc}
\usepackage{newtxtext}

\usepackage{xcolor}
\usepackage{amsmath,amssymb,amsfonts}
\usepackage{graphicx}
\usepackage{booktabs}
\usepackage{makecell}
\usepackage[normalem]{ulem}
\usepackage{cancel}
\usepackage{diagbox}
\usepackage{amsthm}

\usepackage{stfloats}

\usepackage{pgfplots}
\pgfplotsset{compat=1.18}
\usepackage{tikz}
\usetikzlibrary{arrows.meta,positioning,calc,fit,shapes.geometric,backgrounds}

\usepackage[colorlinks=true,linkcolor=blue,citecolor=blue,urlcolor=blue]{hyperref}
\usepackage{cleveref}
\usepackage{microtype}

\newtheorem{theorem}{Theorem}

\theoremstyle{definition}
\newtheorem{remark}{Remark}

\newenvironment{keywords}
  {\par\smallskip\noindent\small\textbf{Keywords:}\ }
  {\par\medskip}

\title{Practical Algebraic Parameter Estimation\\for Noisy Data via Gaussian Process Regression}

\author{
    Oren Bassik\thanks{CUNY Graduate Center, Ph.D. Program in Mathematics, 365 Fifth Avenue,
    New York, NY 10016, USA (e-mail: obassik@gradcenter.cuny.edu).},
    Alexander Demin\thanks{Laboratoire d'informatique de l'{\'E}cole polytechnique, LIX, UMR 7161, CNRS, 1 rue Honor{\'e} d'Estienne d'Orves, 91120 Palaiseau, France (e-mail: demin@lix.polytechnique.fr).}, and
    Alexey Ovchinnikov\thanks{CUNY Queens College, Department of Mathematics,
    65-30 Kissena Blvd, Queens, NY 11367, USA and 
    CUNY Graduate Center, Ph.D. Programs in Mathematics and Computer Science, 365 Fifth Avenue,
    New York, NY 10016, USA (e-mail: aovchinnikov@qc.cuny.edu).}~
    \thanks{This work was partially supported by the NSF grants CCF-2212460 and DMS-1853650. This work was supported by an ERC-2023-ADG grant for the ODELIX project (number 101142171).}
}

\date{}

\begin{document}
\maketitle

\begin{figure*}[!b]
\footnotesize
\noindent
\begin{minipage}[c]{0.84\textwidth}
\itshape
Funded by the European Union. Views and opinions expressed are however
those of the author(s) only and do not necessarily reflect those of the
European Union or the European Research Council Executive Agency.
Neither the European Union nor the granting authority can be held
responsible for them.
\end{minipage}
\hfill
\begin{minipage}[c]{0.15\textwidth}
\centering
\includegraphics[height=0.86cm]{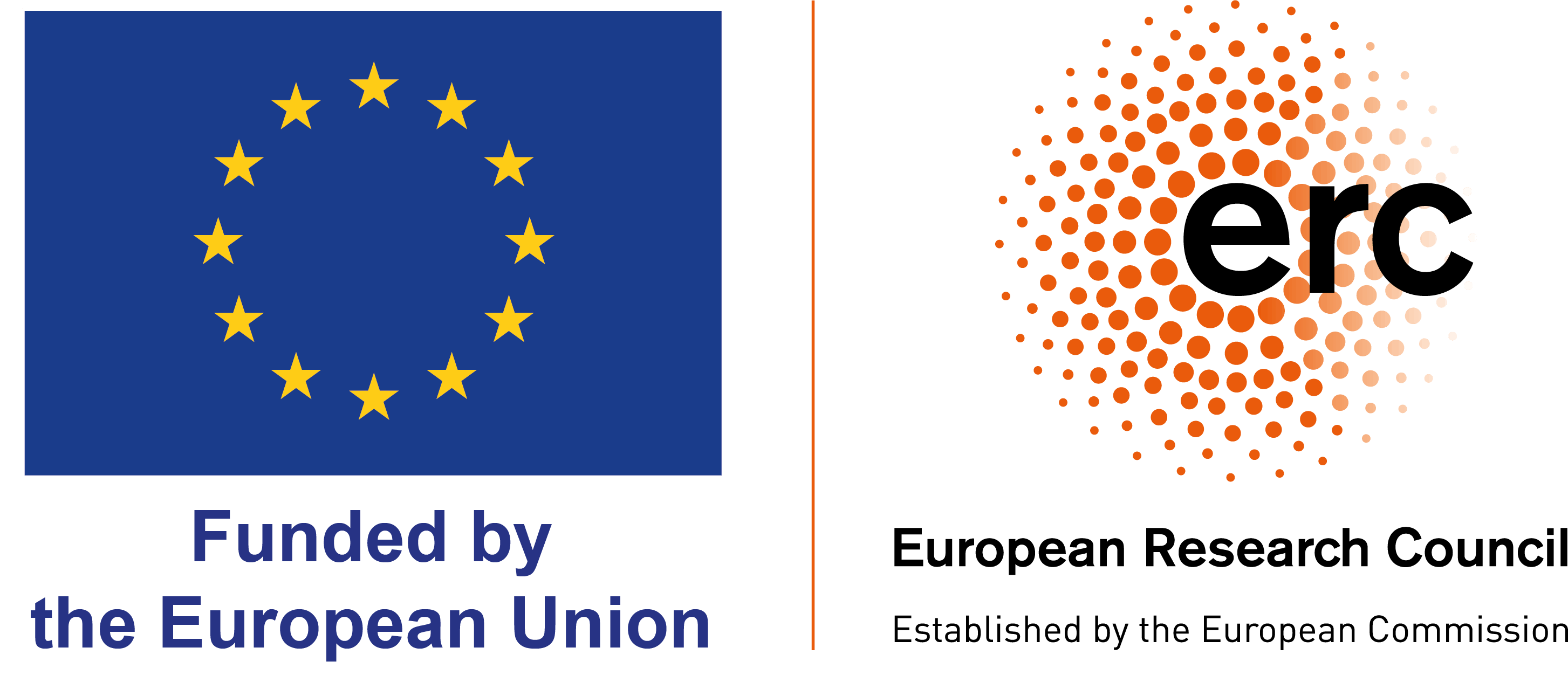}
\end{minipage}
\end{figure*}

\begin{abstract}
Parameter estimation for ordinary differential equation (ODE) models is a fundamental task that is often complicated by the limitations of conventional optimization-based methods. In theory, differential-algebraic approaches offer an appealing alternative: they reduce the problem to polynomial system solving and do not require user-supplied initial guesses for parameter values.
In practice, however, algebraic methods have been limited by their sensitivity to measurement noise, because they require accurate derivatives of observed outputs. In this work, we integrate Gaussian Process Regression (GPR) into the differential-algebraic method and derive a first-order error analysis in terms of noise level and algebraic sensitivity. We evaluate the method across several noise levels on a benchmark of 25 dynamical systems arising in applications including mechanical engineering and systems biology. The proposed method achieves the highest aggregate performance among the methods considered, recovering all sought parameter values and initial conditions to within 10\% relative error in 88.5\% of runs. These results demonstrate that robust derivative estimation can make differential-algebraic parameter estimation practical for dense, noisy synthetic data while retaining key advantages of the algebraic formulation.
% Parameter estimation for ordinary differential equation (ODE) models is a critical task often hindered by noisy data and the limitations of traditional optimization methods. The differential-algebraic approach offers theoretical advantages: it can generate candidate parameter sets from polynomial systems induced by output derivatives, without requiring user-supplied initial parameter guesses.
% In practice, however, algebraic methods have been limited by their sensitivity to measurement noise, because they require accurate derivatives of observed outputs. In this work, we integrate Gaussian Process Regression (GPR) into the derivative estimation stage of a differential-algebraic pipeline and derive a local bound expressing parameter and state error in terms of derivative estimation error and algebraic sensitivity. We evaluate the resulting method on a comprehensive benchmark of 1{,}250 datasets spanning 25 dynamical systems and five noise levels. When paired with a local refinement step, the proposed method achieves the highest aggregate benchmark performance among the compared methods, with 79.6\% success at a 1\% error tolerance. These results show that robust derivative estimation can make differential-algebraic parameter estimation practical for dense, noisy synthetic time-series data while preserving key advantages of the algebraic formulation.
\end{abstract}


\section{Introduction}
\label{sec:introduction}
Parameter estimation for nonlinear systems of ordinary differential equations (ODEs) is a fundamental challenge 
across
science and engineering. 
While ODE models are ubiquitous, their practical use requires precise values for unknown parameters that must be inferred from experimental data. 
The dominant paradigm for this task is nonlinear optimization, where parameter values are sought that minimize the discrepancy between observed measurements and a model's simulated output. 
% obtained by numerical integration of the ODE system. 
This problem is challenging due to its nonconvexity and potential ill-conditioning~\cite{Raue2015,Villaverde2014}.
Optimization methods are also susceptible to well-known practical limitations: they often require a search region known to contain the sought parameter values, together with carefully chosen initial guesses or multistart strategies, to avoid converging to non-optimal local minima.

\begin{figure}[!htpb]
  \centering
  \resizebox{\columnwidth}{!}{\input{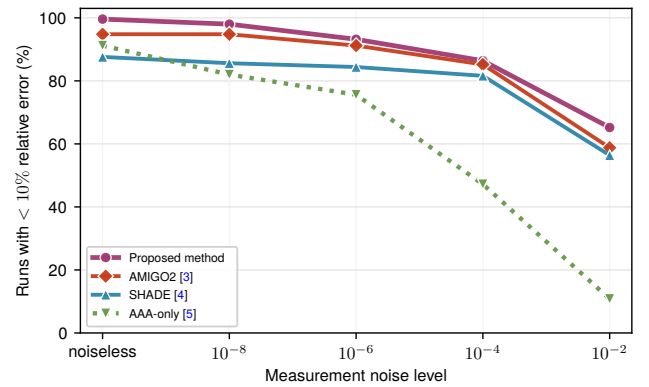}}
  \caption{Success rate as measurement noise increases, aggregated over 25 benchmark examples. A run is counted as successful when every structurally identifiable parameter and initial condition is recovered with relative error below 10\%. The proposed method includes local refinement; the AAA-only method uses the same algebraic pipeline restricted to AAA rational interpolation, without refinement.}
  \label{fig:sr10-noise-all-arms}
\end{figure}

The differential-algebraic approach is an alternative method, {based on the theory from~\cite{Ljung1994,hong_sian_2019}}, which transforms the ODE system into a set of algebraic equations~\cite{bassik2023robustparameterestimationrational,Demin2025}. Then the problem can be expressed as a polynomial system and solved using methods from numerical algebraic geometry. In principle, this enables recovery of all parameter and initial condition values without requiring search regions, initial parameter guesses, or repeated numerical integration of the ODE system.
However, the approach requires calculating high-order derivatives of measured outputs, which is highly sensitive to measurement noise. Consequently, the practical usability of differential-algebraic methods is often limited by data quality.

A standard approach to dealing with derivative estimation from noisy observations is an intermediate smoothing or regularization step. Some classical approaches include smoothing splines and related least-squares formulations, which yield differentiable approximations of the measured signals and have been extensively used in the analysis of dynamical systems~\cite{Ramsay2007, Cao2009}. Kernel and local polynomial regression methods have also been proposed to estimate both states and their derivatives from noisy data~\cite{Liang2008}, forming the basis of two-stage estimation procedures in which a smoothed trajectory is first obtained and subsequently used for parameter estimation~\cite{Brunel2008}.

In this paper, we consider the problem of joint estimation of model parameters and unknown initial conditions. We integrate statistical smoothing methods into the differential-algebraic parameter estimation framework, and we strengthen the algebraic framework itself.
Specifically, we
use Gaussian Process Regression (GPR), a standard tool from statistical machine learning, as a differentiable smoother inside a differential-algebraic parameter estimation pipeline. While the foundational algebraic framework in~\cite{bassik2023robustparameterestimationrational} was benchmarked on noise-free data, in this paper we focus on dense, noisy measurements.
% the central contribution of this paper is demonstrating that integrating robust derivative estimation makes the method effective on dense, noisy measurements, a claim we validate with an extensive benchmark.

Gaussian processes have also been used more directly for ODE inference, including Bayesian GP-ODE and gradient-matching formulations~\cite{Girolami2008,Calderhead2009}, as well as recent GP/NeuralODE and variational multiple-shooting approaches~\cite{Bhouri2022,Hegde2022}. Our use of GPR is narrower: independent GP regressions serve as differentiable noise-aware smoothers inside a differential-algebraic pipeline, rather than as a full Bayesian inference model over states, dynamics, or parameters.

% The contributions of this paper are:
% \begin{itemize}
%     \item A practical differential-algebraic parameter estimation framework for noisy data by combining regression-based derivative estimation with several robustness-oriented extensions, including deterministic subsystem selection, aggregation across shooting configurations and derivative estimators, and optional bounded local polishing. 
%     \item A reproducible large-scale empirical evaluation on 25 ODE models spanning multiple noise levels, showing that these modifications make differential-algebraic parameter estimation competitive with established optimization-based approaches such as AMIGO2~\cite{AMIGO2} and SHADE on dense noisy measurements.
% \end{itemize}

The main contributions of this paper are:
\begin{enumerate}
    \item A differential-algebraic parameter estimation pipeline for dense noisy data that combines GPR-based derivative estimation with aggregation across derivative estimators and shooting configurations and optional local refinement.
    \item A deterministic construction of the square polynomial subsystem, replacing the arbitrary subsystem choice in~\cite{bassik2023robustparameterestimationrational}. The construction first minimizes the highest required output derivative order and then performs a bounded heuristic search guided by mixed volume, reducing reliance on difficult derivative estimates before reducing the cost of the polynomial solve.
    \item An error analysis for a fixed GPR fit and selected subsystem. It separates reconstruction bias from propagated measurement noise, gives a probabilistic local bound for the resulting parameter and state error, and quantifies kernel-dependent noise amplification for the SE and RQ kernels.
    \item A reproducible comparison with AMIGO2~\cite{AMIGO2} and SHADE on 25 systems at five noise levels, together with controlled comparisons of the full derivative estimator suite against AAA alone and of polished against unpolished estimates.
\end{enumerate}

This paper is organized as follows. In~\Cref{sec:background}, we recall some relevant background. In~\Cref{sec:methodology}, we describe our proposed method. In Section~\ref{sec:theory}, we assess how the accuracy of the proposed parameter estimation method depends on the error of derivative estimation and the choice of the polynomial subsystem. We illustrate the method on a controlled CSTR example in~\Cref{sec:examples}. In~\Cref{sec:experimental_setup,sec:results,sec:discussion}, we present the experimental setup, results, and discussion. We conclude in~\Cref{sec:conclusion}.

\section{Background}
\label{sec:background}

\subsection{Model parametrization}
\label{ssec:model_parametrization}
Consider a system given in the conventional state space form of the dynamics between the input and the output:
\begin{align}
&\mathbf{x}'(t) = \mathbf{f}(\mathbf{x}(t), \mathbf{u}(t), \mathbf{p}),\label{eq:x_state_dynamics}\\
&\mathbf{y}(t) = \mathbf{g}(\mathbf{x}(t), \mathbf{u}(t), \mathbf{p}),\label{eq:y_output}
\end{align}
where $\mathbf{x}(t)$ are the states or internal variables, $\mathbf{u}(t)$ are the input signals, $\mathbf{p}$ are the time-independent parameters, and $\mathbf{y}(t)$ are the observed outputs.

In our experimental setup, the continuous input functions $\mathbf{u}(t)$ are prescribed and known (i.e., they do not depend on the unobserved state). The unknown quantities are the initial conditions $\mathbf{x}(0)$ and the parameter values $\mathbf{p}$. It is also assumed that for every observed output $y_i$ we have access to the measured output data $(t_{i,1}, y_i(t_{i,1})), \ldots, (t_{i,N_i}, y_i(t_{i,N_i}))$. Then the task of parameter estimation is to reconstruct the unknown values of $\mathbf{x}(0)$ and $\mathbf{p}$ from this data.

\subsection{The Differential-Algebraic Approach to Parameter Estimation}
\label{ssec:diff_alg_approach}
In the differential-algebraic setting, we consider models of the form~\eqref{eq:x_state_dynamics} and~\eqref{eq:y_output} with rational dynamics, that is, with $\mathbf{f}$ and $\mathbf{g}$ being rational functions of states, inputs, and parameters.
A detailed exposition of the method can be found in~\cite[Section~3]{bassik2023robustparameterestimationrational}. Here, we provide an example.

To illustrate the method, we consider the toy ODE system (where we omit explicit dependence on time)
\[
\left.\begin{aligned}
x' = a x^2 + b,\quad~~ y = x,
\end{aligned}\right.
\]
where $a,b$ are unknown time-independent parameters. We also have access to the data $(t_i, y(t_i)), i=1,2,3,4$ observed in an experiment: 
\[
\{(0.00, 1.00), ~(0.33, 1.42), ~(0.67, 2.18), ~(1.00, 4.14)\}.
\]
We used the initial condition $x(0) = 1.00$ and the values $a = 0.60$ and $b = 0.40$ to simulate the ODE and obtain this data. The goal of parameter estimation is to recover these values.

By differentiating the output equation twice (in general, the required order of differentiation is made precise in~\cite{hong_sian_2019}), we construct the system:
\begin{equation*}
    \label{eq:poly-sys2}
    \begin{aligned}
        y &= x \\
        y' &= x'\\
        y'' &= x''\\
    \end{aligned}
    \quad\quad
    \begin{aligned}
        x' &= a x^2 + b\\
        x'' &= 2 a x x'\\
    \end{aligned}
\end{equation*}
The next step is to estimate the output $y$ and its time derivatives at a chosen time point from the data $(t_i, y(t_i))$. 
% using an appropriate derivative estimator; this is also a central concern of the present paper. 
For this example, suppose at $t = 0$ we obtain the estimates:
\[
\hat{y}(0) \approx 1.00, \quad \hat{y}'(0) \approx 1.00, \quad \hat{y}''(0) \approx 1.20.
\]

Substituting these values into the differentiated equations at $t = 0$ gives a polynomial system in the indeterminates $a, b, x_0, x'_0, x''_0$ (for brevity, we denote $x_0 = x(0)$):
\begin{equation*}
    \begin{aligned}
        1.00 &= x_0 \\
        1.00 &= x'_0\\
        1.20 &= x''_0\\
    \end{aligned}
    \quad\quad
    \begin{aligned}
        x'_0 &= a x_0^2 + b\\
        x''_0 &= 2 a x_0 x'_0\\
    \end{aligned}
\end{equation*}
Solving this system recovers the parameters together with the initial condition,
\[(a, b, x_0) = (0.60,\ 0.40,\ 1.00),\]
which match the values we used to generate the data. In general the polynomial system may admit several solutions, and additional criteria, such as known sign constraints or fit quality, are used to select among them.

\subsection{Gaussian Process Regression for Derivative Estimation}
\label{ssec:gpr_for_derivatives}

We provide a brief overview of Gaussian Process Regression, focusing on properties useful for our method; for a thorough treatment, see~\cite{Rasmussen2006}. A Gaussian Process (GP) is a collection of random variables, any finite number of which have a joint Gaussian distribution. A GP prior over a time-dependent latent function is fully specified by its mean function $m(t)$ and covariance function (or kernel) $k(t, t')$.
% which we write as $f(t) \sim \mathcal{GP}(m(t), k(t, t'))$.

In a regression context, we assume that our noisy observations $y(t_i)$ at inputs $t_i$ are generated from a latent function $f(t)$ corrupted by Gaussian noise: $y_i = f(t_i) + \epsilon_i$, where $\epsilon_i \sim \mathcal{N}(0, \sigma_\epsilon^2)$.

Given $n$ observations $Y = (y(t_1), \ldots, y(t_n))^\top$ at times $t_1, \ldots, t_n$, the posterior predictive distribution at a test point $t_*$ is Gaussian with mean and variance
\begin{align}
    \hat{f}(t_*) &= m(t_*) + \mathbf{k}_*^\top \bigl(\mathbf{K} + \sigma_{\mathrm{GP}}^2 \mathbf{I}\bigr)^{-1} \bigl(Y - \mathbf{m}\bigr), \label{eq:gp_mean}\\
    \operatorname{Var}[f(t_*)] &= k(t_*, t_*) - \mathbf{k}_*^\top \bigl(\mathbf{K} + \sigma_{\mathrm{GP}}^2 \mathbf{I}\bigr)^{-1} \mathbf{k}_*, \label{eq:gp_var}
\end{align}
where $\mathbf{K}$ is the $n \times n$ matrix with entries $K_{ij} = k(t_i, t_j)$, $\mathbf{k}_* = (k(t_1, t_*), \ldots, k(t_n, t_*))^\top$, and $\mathbf{m} = (m(t_1), \ldots, m(t_n))^\top$, and $\sigma_{\mathrm{GP}}^2$ is a fitted parameter of the GPR model; it need not equal the data-generating noise variance $\sigma_\epsilon^2$. The posterior mean~\eqref{eq:gp_mean} is a smooth, analytic function of $t_*$ that serves as an estimate of the latent noise-free signal.

The choice of kernel is crucial as it encodes our prior assumptions about the function's properties, such as smoothness. For this work, we use two infinitely differentiable kernels. The first is the common squared exponential (or RBF) kernel:
\begin{equation}\label{eq:ker-se}
    k_{\mathrm{SE}}(t, t') = \sigma_f^2 \exp\left(-\frac{(t - t')^2}{2\ell^2}\right),
\end{equation}
defined by the length scale $\ell$, which controls the smoothness or characteristic frequency of the function, and the signal variance $\sigma_f^2$. The second is the rational quadratic kernel:
\begin{equation}\label{eq:ker-rq}
    k_{\mathrm{RQ}}(t, t') = \sigma_f^2 \left(1 + \frac{(t - t')^2}{2\alpha_{\mathrm{RQ}}\ell^2}\right)^{-\alpha_{\mathrm{RQ}}},
\end{equation}
which is a smooth alternative to the squared-exponential kernel that can accommodate variation over more than one time scale. Given the data $(t_i, y(t_i))$, the GPR model parameters, including $\sigma_{\mathrm{GP}}^2$, are learned by maximizing the log marginal likelihood.

For this choice of kernels, the derivative $\hat{f}^{(h)}$ exists and admits a closed-form expression for every $h$~\cite[Ch.~9.4]{Rasmussen2006}. It provides an estimate of the derivative of the noise-free signal.

As with other smoothing methods, in practice, derivative estimates are often less reliable near the boundaries of the observation interval; this motivates the use of multiple shooting points and aggregation in~\Cref{step:multi-time-points}.

\section{Methodology: A GPR-Enhanced Algebraic Framework with Deterministic Subsystem Selection}
\label{sec:methodology}

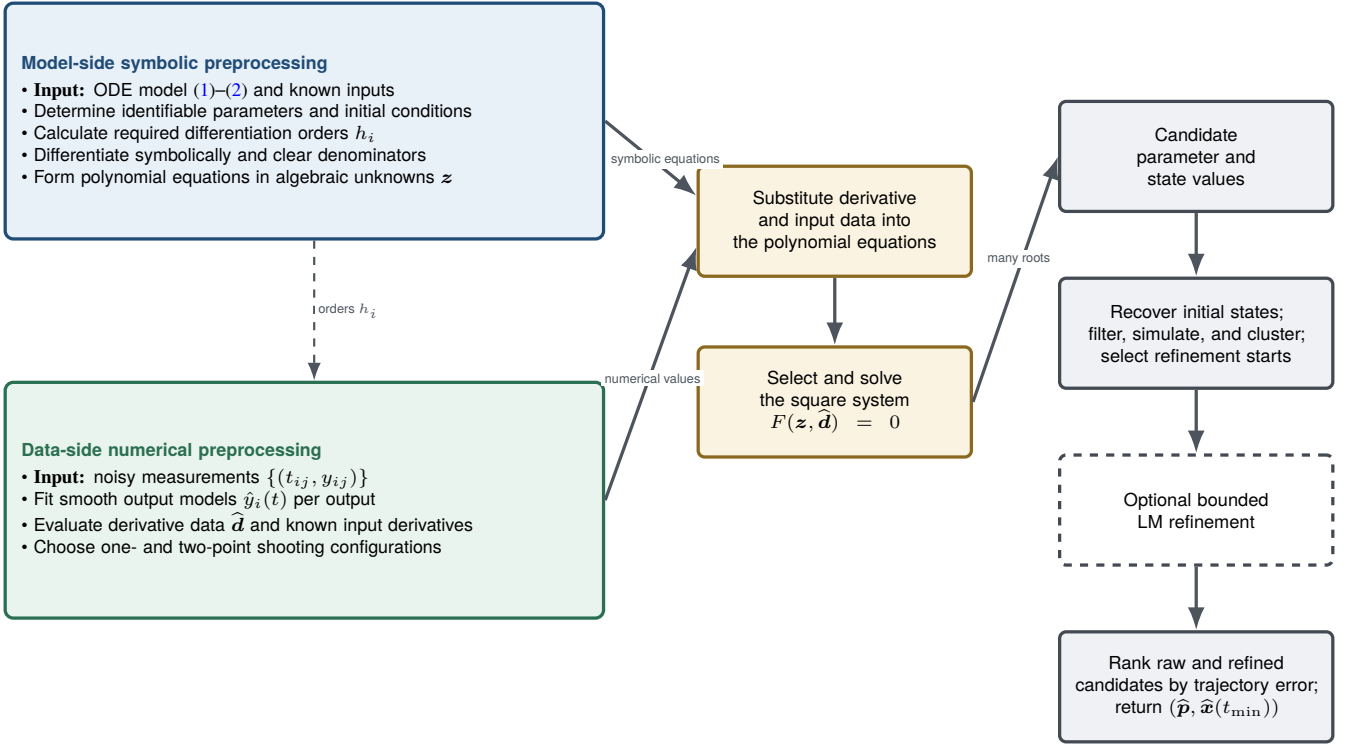
\begin{figure*}[t]
\centering
% Method-overview schematic for Section III.
% Pure TikZ; no generated inset figures are required.
\definecolor{modelblue}{RGB}{32,76,119}
\definecolor{modelbg}{RGB}{232,240,248}
\definecolor{datagreen}{RGB}{43,111,82}
\definecolor{databg}{RGB}{232,244,237}
\definecolor{solvegold}{RGB}{139,104,31}
\definecolor{solvebg}{RGB}{250,243,225}
\definecolor{postgray}{RGB}{68,74,83}
\definecolor{postbg}{RGB}{239,241,244}
\definecolor{linegray}{RGB}{77,84,93}

\resizebox{\textwidth}{!}{%
\begin{tikzpicture}[
  x=1in,y=1in,
  >=Latex,
  panel/.style={
    draw=#1,
    fill=#1!7,
    rounded corners=3pt,
    very thick,
    align=left,
    font=\sffamily\scriptsize,
    text width=2.88in,
    minimum height=1.20in,
    inner sep=6pt
  },
  modelpanel/.style={panel=modelblue, fill=modelbg},
  datapanel/.style={panel=datagreen, fill=databg},
  box/.style={
    draw=#1,
    fill=#1!8,
    rounded corners=2.4pt,
    very thick,
    align=center,
    font=\sffamily\scriptsize,
    text width=1.28in,
    minimum height=0.56in,
    inner sep=4pt
  },
  solvebox/.style={box=solvegold, fill=solvebg},
  postbox/.style={box=postgray, fill=postbg},
  optionalbox/.style={
    draw=postgray,
    fill=white,
    rounded corners=2.4pt,
    very thick,
    dashed,
    align=center,
    font=\sffamily\scriptsize,
    text width=1.28in,
    minimum height=0.56in,
    inner sep=4pt
  },
  arrow/.style={->,very thick,linegray},
  softarrow/.style={->,thick,linegray,dashed},
  lab/.style={font=\sffamily\tiny,text=linegray,fill=white,inner sep=1.2pt}
]

\node[modelpanel] (model) at (1.58,3.13) {
  {\bfseries\color{modelblue} Model-side symbolic preprocessing}\\[2pt]
  \textbullet\ {\bf Input:} ODE model~\eqref{eq:x_state_dynamics}--\eqref{eq:y_output} and known inputs\\
  \textbullet\ Determine identifiable parameters and initial conditions\\
  \textbullet\ Calculate required differentiation orders $h_i$\\
  \textbullet\ Differentiate symbolically and clear denominators\\
  \textbullet\ Form polynomial equations in algebraic unknowns $\boldsymbol z$
};

\node[datapanel] (data) at (1.58,1.20) {
  {\bfseries\color{datagreen} Data-side numerical preprocessing}\\[2pt]
  \textbullet\ {\bf Input:} noisy measurements $\{(t_{ij},y_{ij})\}$\\
  \textbullet\ Fit smooth output models $\hat y_i(t)$ per output\\
  \textbullet\ Evaluate derivative data $\widehat{\boldsymbol d}$ and known input derivatives\\
  \textbullet\ Choose one- and two-point shooting configurations
};

\node[solvebox] (instantiate) at (4.28,2.62) {
  Substitute derivative\\and input data into\\the polynomial equations
};
\node[solvebox] (roots) at (4.28,1.70) {
  Select and solve\\the square system\\[-1pt]
  {\scriptsize $F(\boldsymbol z,\widehat{\boldsymbol d})=0$}
};

\node[postbox] (candidates) at (6.12,2.95) {
  Candidate\\parameter and\\state values
};
\node[postbox] (rank) at (6.12,2.05) {
  Recover initial states;\\filter, simulate, and cluster;\\select refinement starts
};
\node[optionalbox] (polish) at (6.12,1.15) {
  Optional bounded\\LM refinement
};
\node[postbox] (output) at (6.12,0.25) {
  Rank raw and refined\\candidates by trajectory error;\\return $(\widehat{\boldsymbol p},\widehat{\boldsymbol x}(t_{\min}))$
};

\draw[arrow] (model.east) -- node[lab,pos=0.66,above] {symbolic equations} ([yshift=0.13in]instantiate.west);
\draw[arrow] (data.east) -- node[lab,pos=0.52,below] {numerical values} ([yshift=-0.13in]instantiate.west);
\draw[softarrow] ($(model.south)+(0.05,0)$) -- node[lab,right] {orders $h_i$} ($(data.north)+(0.05,0)$);

\draw[arrow] (instantiate) -- (roots);
\draw[arrow] (roots.east) -- node[lab,pos=0.54,above] {many roots} (candidates.west);
\draw[arrow] (candidates) -- (rank);
\draw[arrow] (rank) -- (polish);
\draw[arrow] (polish) -- (output);

\end{tikzpicture}%
}
\caption{Overview of the proposed method.}
\label{fig:method-overview}
\end{figure*}

\Cref{fig:method-overview} summarizes the symbolic and numerical objects used by the estimator. We now describe each step of the general procedure.

\subsection{Step 1: Structural Identifiability Analysis}\label{method:step1}
Before attempting to estimate parameters, we first perform a structural identifiability analysis on the ODE model~\eqref{eq:x_state_dynamics}-\eqref{eq:y_output} using methods from differential algebra~\cite{Ljung1994,Raue2009}, as implemented in \texttt{SIAN.jl}~\cite{hong_sian_2019} and \texttt{StructuralIdentifiability.jl}~\cite{stident,demin2026simplegeneratorsrationalfunction}. Structural identifiability analysis determines, from the model structure alone, whether it is theoretically possible to uniquely determine the parameters from noise-free data. It is a necessary condition on the model structure rather than a guarantee of accurate recovery from finite, noisy, partially observed data. 

This analysis provides two pieces of information used by the subsequent steps: 
(1) the differentiation order $h$ required by the differential-algebraic method and (2) the set of parameters and initial conditions that are structurally unidentifiable. The unidentifiable quantities are excluded from our analysis.

\subsection{Step 2: Symbolic Differentiation}\label{method:step2}
Using the differentiation order determined in Step~\hyperref[method:step1]{1}, we differentiate the equations of the ODE system symbolically with respect to time. As illustrated in Section~\ref{ssec:diff_alg_approach}, this generates a set of polynomial equations in the indeterminates $\mathbf{p}, \mathbf{x}, \mathbf{u}$, $\mathbf{y}$ and their time derivatives, relating the system parameters to the state, input, and output variables. We reduce the size of the system by eliminating the unidentifiable quantities from it based on structural identifiability analysis.
When rational terms occur in the model, denominators are cleared at this stage before the equations are passed to the polynomial solver.

\subsection{Step 3: Signal and Derivative Estimation}\label{method:step3}
Let ${\bf y} = (y_1,\ldots,y_m)$ denote the observed noisy outputs of the model. Given noisy measurements $\{(t_{i,j}, y_{i}(t_{i,j}))\}_{j=1}^{n_i}$, the goal of this step is to reconstruct a smooth approximation $\hat{y}_i(t)$ of the latent noise-free signal for each $i=1,\ldots,m$.

The differential-algebraic framework requires estimates of output derivatives at selected time points. Direct differentiation of noisy measurements is highly unstable. In this work, we use Gaussian Process Regression (GPR) to obtain smooth approximations of the observed outputs and their derivatives.

For each $i=1,\ldots,m$, a Gaussian process model defined by~\eqref{eq:gp_mean}-\eqref{eq:gp_var} is fitted to the noisy measurements, yielding a posterior mean function $\hat{y}_i(t)$. The hyperparameters, including the observation-noise variance, are learned by maximizing the log marginal likelihood.
For these GPR fits, we consider both kernels~\eqref{eq:ker-se} and~\eqref{eq:ker-rq}. Rather than selecting a single kernel a priori, each is fitted independently and contributes its own candidate solutions downstream, arbitrated by the solution validation of Step~\hyperref[method:step7]{7}. %Exact GP regression costs $O(n_i^3)$ time per fit, which is modest at our sample size ($n_i=750$); fitted models are cached and shared across estimators. % arms that use the same kernel.

Because these kernels are infinitely differentiable, derivatives of any required order can be computed directly from the posterior mean function; in the implementation, they are evaluated by Taylor-mode automatic differentiation of $\hat{y}_i(t)$. For every $i=1,\ldots,m$, with $h$ being the required differentiation order computed in Step~\hyperref[method:step1]{1}, the resulting estimates
\[
\hat{y}_i(t),\ldots,\hat{y}_i^{(h)}(t)
\]
serve as input to the polynomial system construction in Step~\hyperref[method:step4]{4}.

The framework is modular in the choice of derivative estimator and can accommodate methods other than GPR, as also noted in~\cite{bassik2023robustparameterestimationrational}. In the implementation used in the benchmarks, we consider a suite of derivative estimators, including GPR-based smoothers, AAA rational interpolation, and Chebyshev approximation variants, with noise-based filtering of interpolation methods expected to fail at higher noise.
Candidate solutions from the admitted estimators are aggregated and ranked downstream.

% This step is the core of our contribution, integrating robust smoothing and derivative estimation into the original algebraic workflow. 
% The implementation is modular in the choice of estimator: any method that produces a smooth approximation of each output together with its derivatives can be used. 

% For noisy data, GPR-family estimators are central in this work because they are much more robust to noise than direct AAA rational interpolation (see \Cref{ssec:aaa_gpr}). In the benchmark implementation, the package evaluates a configured suite of derivative estimators, including GPR-based smoothers, AAA rational interpolation, and Chebyshev approximation variants, with noise-based filtering of interpolation methods known to fail at higher noise. Candidate solutions from the admitted estimators are aggregated and ranked downstream.

% With the kernels we have chosen, derivatives of any order can be obtained directly, which is the key property that makes GPR well suited to the differential-algebraic framework, where accurate high-order derivatives are the essential input.

\subsection{Step 4: Polynomial System Formulation}\label{method:step4}
This step combines the symbolic equations obtained in Step~\hyperref[method:step2]{2} with the derivative estimates produced in Step~\hyperref[method:step3]{3}.

A time point $t_0$ is selected. For each output component $y_i$, we evaluate the fitted output approximation $\hat{y}_i(t)$ and its derivatives up to order $h$ at $t_0$, obtaining
\[\hat{y}_i(t_0),\ldots,\hat{y}_i^{(h)}(t_0).\]
For each known input component $u_i$, similarly the values of the derivatives\[u_i(t_0), \ldots, u_i^{(h)}(t_0)\] are obtained. When $u_i(t)$ is specified as a differentiable function of time, these values are computed automatically; otherwise, they may be supplied directly.

These values are substituted in the polynomial equations generated in Step~\hyperref[method:step2]{2}. Specifically, each occurrence of $y_i^{(j)}$ and $u_i^{(j)}$ is replaced by the corresponding numerical value evaluated at $t_0$.
The result is a numerical multivariate polynomial system whose unknowns are the parameters ${\bf p}$ together with the values of the state variables ${\bf x}$ and their derivatives at $t_0$. 

\subsection{Step 5: Polynomial System Solving}\label{method:step5}

The polynomial system obtained in Step~\hyperref[method:step4]{4} is typically overdetermined. In the noise-free setting, the equations would be consistent. However, because the derivatives are estimated from noisy observations, the system is generally inconsistent. We therefore replace the full system by a carefully chosen square subsystem, 
\[
F(\boldsymbol z,\boldsymbol d)=0,
\]
where $\boldsymbol z$ collects the indeterminates (such as parameters, states, and their derivatives), $\boldsymbol d$ contains the exact output values and derivatives required by the subsystem, and $\widehat{\boldsymbol d}$ contains their estimates from Step~\hyperref[method:step3]{3}. Thus $F(\boldsymbol z,\boldsymbol d)=0$ is the exact-data system, while the numerical system replaces $\boldsymbol d$ by $\widehat{\boldsymbol d}$.
% This system subsequently will be solved using numerical algebraic geometry methods.

A straightforward approach is to form random linear combinations of the original equations. In practice, however, the equations contain derivatives of different orders and magnitudes. Combining such equations can amplify estimation errors. The method of~\cite{bassik2023robustparameterestimationrational} instead selects a zero-dimensional square subsystem using a Jacobian rank criterion. 

In this work, we select the square subsystem by a deterministic rule aimed at robustness. The rule first minimizes the order of the output derivatives the subsystem requires, since higher-order derivative estimates can be more sensitive to noise. At that order, it performs a bounded heuristic search over full-rank square subsystems and selects the candidate of smallest computed mixed volume, the generic number of solutions and hence the number of homotopy paths the solver must track. The selected system is solved by homotopy continuation~\cite{HomotopyContinuation.jl}, and its real solutions are retained as candidate values for the parameters ${\bf p}$ and the state variables ${\bf x}(t_0)$ at time $t_0$.

% This system is typically overdetermined, so we reduce it to a square system that can be solved readily with numerical algebraic geometry tools. 
% A straightforward approach is to select the equations of the square system as random linear combinations of the original equations, but this is brittle and exacerbates sensitivity to inexact derivative estimates.
% In~\cite{bassik2023robustparameterestimationrational}, the authors pick an arbitrary square subsystem of dimension zero by using a Jacobian rank test. In the method in this paper, we developed the following approach, which we have found to be more efficient.

\subsection{Step 6: Aggregation Across Time Points and Derivative Estimators}\label{step:multi-time-points}
In practice, estimation at a single time point $t_0$ can be sensitive to the local data quality. To improve robustness, we solve the estimation problem in several configurations and aggregate the resulting candidate parameter sets. We use two kinds of configurations: \emph{single-point} estimations, which use the derivative estimates at one time point as described above, and \emph{two-point} estimations, which jointly use the derivative estimates at a pair of time points, coupling the two corresponding systems through the shared parameters. For the benchmarks in this paper we use 20 single-point estimations at exponentially warped shooting points: for $r=1,\ldots,20$, we set $u_r=(r-1)/19$ and $s_r=(e^{3u_r}-1)/(e^3-1)$, then take the sampled time point nearest to $t_{\min}+s_r(t_{\max}-t_{\min})$. This clusters more points near the initial time. Although this can place some shooting points where derivative estimates are less reliable (\Cref{ssec:gpr_for_derivatives}), the aggregation across configurations and the validation in Step~\hyperref[method:step7]{7} are designed to absorb such per-configuration unreliability. We also use up to 15 two-point estimations per dataset, chosen from these shooting points by largest temporal separation.
In the benchmark implementation, these shooting configurations are evaluated with the suite of derivative estimators described in Step~3, subject to noise-based estimator filtering. Each estimator/configuration pair yields a collection of candidate parameter sets, and the process is naturally parallelizable since the configurations are independent. The resulting candidate sets are aggregated before the final filtering step.
When a candidate is obtained at a later shooting time, its estimated state is integrated backward to the first observation time to obtain the reported initial condition.

\subsection{Step 7: Solution Filtering and Validation}\label{method:step7}
The aggregated set of candidate solutions is first filtered to remove any solutions that are non-physical (e.g., negative parameter values if known to be positive) or, optionally, lie outside user-provided bounds. For the remaining candidates, we perform a full numerical simulation of the original ODE system and score the candidate by the trajectory sum of squared errors
\[
    \sum_{\ell}\sum_j
    \left(y_{\ell,\operatorname{model}}(t_j)-y_{\ell,\operatorname{data}}(t_j)\right)^2,
\]
where the outer sum is over measured outputs. Candidates are clustered and ranked by this score. In the polished variant, the polishing starts are chosen after this filtering and clustering, and the final raw-and-polished candidate pool is ranked by the same score.

\subsection{Step 8: Optional Polishing Step}
In addition to the raw algebraic solutions, we consider a polished variant in which the algebraic solution is used as an initial guess for a local least-squares refinement. For this step, we use bounded Levenberg--Marquardt least squares in per-variable log coordinates, with the same trajectory-error score used for candidate ranking.

\section{Error Analysis}
\label{sec:theory}

The goal of this section is to prove~\Cref{thm:fixed-design-local}, which analyzes the error in the GPR and algebraic stages of the proposed method.
We study the GPR and algebraic stages (Steps~\ref{method:step3}-\ref{method:step5}) for a single shooting point and a fixed selected subsystem. The theorem separates the error of derivative estimation from the sensitivity of the constructed polynomial system to that error. We use this analysis to interpret the empirical behavior reported in \Cref{sec:results}, and we return to that connection in \Cref{sec:discussion}.

Our bounds are probabilistic in the following model (see also~\Cref{rem:scope}): the measurement noise is random; the measurement times, the GPR mean and kernel hyperparameters, the shooting point, and the selected square subsystem are treated as fixed and do not depend on the measurement noise.

\subsection{Derivative Error for a Fixed GPR Fit}

Let $f_i$ denote the noise-free output functions, and let $Y \in \mathbb{R}^N$ collect the noisy measurements $y_{i,j}=f_i(t_{i,j})+\epsilon_{i,j}$ for $i=1,\ldots,m$ and $j=1,\ldots,n_i$.
Write
\begin{equation}
  Y=Y^\star+\epsilon,
  \qquad \epsilon\sim\mathcal N(0,\Sigma_\epsilon),
  \label{eq:fixed-design-observations}
\end{equation}
where $Y^\star$ contains the corresponding noise-free output values.
As in~\cite{Liang2008,Hasenauer2017}, we assume that the measurement errors are independent across observation times and outputs. Consequently, $\Sigma_\epsilon$ is block diagonal, with block $\sigma_{\epsilon,i}^2\boldsymbol I_{n_i}$ for the $i$-th output, where $\sigma_{\epsilon,i}$ is its measurement noise standard deviation.

Let $q$ be the number of output values and derivatives in $\boldsymbol d$, and let $\widehat{\boldsymbol d}(Y)\in\mathbb R^q$ stack their GPR estimates in the same order. We first record that, for a fixed fit, every component of $\widehat{\boldsymbol d}(Y)$ is an affine function of the data $Y$. Indeed, for one output with measurement block $Y_i$, let $\boldsymbol K$ and $\boldsymbol m_i$ denote its training covariance matrix and prior-mean vector, and let $\boldsymbol k_t=(k(t_{i,1},t),\ldots,k(t_{i,n_i},t))^{\mathsf T}$. For any $j\geqslant 0$, differentiating the posterior mean~\eqref{eq:gp_mean} $j$ times with respect to time gives
\[
  \widehat f_i^{(j)}(t)
  =m_i^{(j)}(t)+
  \left(\partial_t^j\boldsymbol k_t\right)^{\mathsf T}
  \left(\boldsymbol K+\sigma_{\mathrm{GP}}^2\boldsymbol I\right)^{-1}
  (Y_i-\boldsymbol m_i).
\]
Every quantity on the right-hand side except $Y_i$ is fixed with respect to our probability model, so this component is affine in $Y_i$, with coefficient row
\begin{equation}
  \boldsymbol w_{j,t}^{\mathsf T}
  =
  \left(\partial_t^j\boldsymbol k_t\right)^{\mathsf T}
  \left(\boldsymbol K+\sigma_{\mathrm{GP}}^2\boldsymbol I\right)^{-1}.
  \label{eq:fixed-design-weight-row}
\end{equation}
Stacking these coefficient rows in the same order as the components of $\widehat{\boldsymbol d}$ defines a matrix $W\in\mathbb R^{q\times N}$. We index its rows by $k=1,\ldots,q$, with each $k$ corresponding to one triple $(i,j,t)$; the $k$-th row acts on the measurement block $Y_i$ through $\boldsymbol w_{j,t}^{\mathsf T}$. Equivalently,
\[
  \widehat{\boldsymbol d}(Y)
  =\widehat{\boldsymbol d}(Y^\star)+W(Y-Y^\star).
\]

Even in the absence of measurement noise, the fixed fit does not reproduce the required output values and derivatives exactly. Applying the same fixed fit to the noise-free samples defines the \emph{reconstruction bias}
\[
  \boldsymbol b:=\widehat{\boldsymbol d}(Y^\star)-\boldsymbol d,
\]
whose component for derivative order $j$ of output $i$ at the point $t$ is
\begin{equation}
  b_{i,j,t}
  :=m_i^{(j)}(t)+
  \boldsymbol w_{j,t}^{\mathsf T}(Y_i^\star-\boldsymbol m_i)
  -f_i^{(j)}(t).
  \label{eq:fixed-design-bias-coordinate}
\end{equation}
This bias is generally nonzero: with $\sigma_{\mathrm{GP}}^2>0$, the posterior mean smooths the data rather than interpolating it, and even exact function values on a finite grid do not determine the derivatives of $f$. Substituting $Y=Y^\star+\epsilon$ into this affine identity yields the exact decomposition
\begin{equation}
  \widehat{\boldsymbol d}-\boldsymbol d
  =\boldsymbol b+W\epsilon,
  \qquad
  \Sigma_d:=W\Sigma_\epsilon W^{\mathsf T}.
  \label{eq:fixed-design-derivative-error}
\end{equation}
Since $\mathbb E[\epsilon]=0$, the reconstruction bias $\boldsymbol b=\mathbb E[\widehat{\boldsymbol d}-\boldsymbol d]$ is the mean error of the fixed estimator under~\eqref{eq:fixed-design-observations}. The matrix $\Sigma_d$ is the covariance induced in the derivative estimates by measurement noise; it is distinct from the conditional GP posterior covariance in~\eqref{eq:gp_var}. Moreover, since $\epsilon$ is Gaussian,
\[
  \widehat{\boldsymbol d}-\boldsymbol d
  \sim\mathcal N(\boldsymbol b,\Sigma_d).
\]

Related derivative error bounds for randomly sampled observation times are established by Liu and Li~\cite{liu-li-2023,liu-li-gp-2025}. For fixed hyperparameters, the kernel ridge estimator studied in their 2023 paper is algebraically identical to the GPR posterior mean~\eqref{eq:gp_mean}, while the companion GP paper gives related asymptotic results. Here the observation times are prescribed and the GPR fit is held fixed, so the measurement noise contributions are represented directly by $W\epsilon$ in~\eqref{eq:fixed-design-derivative-error}.

\subsection{Local Error in the Algebraic Solution}

On the algebraic side, recall from Step~\hyperref[method:step5]{5} the selected square system $F(\boldsymbol z,\boldsymbol d)=0$, viewed as a map
\[
  F:\mathbb R^{n_z}\times\mathbb R^q\longrightarrow\mathbb R^{n_z}
\]
whose first argument collects the algebraic unknowns and whose second argument carries the data. Let $P\boldsymbol z$ collect the model parameters and the state values at the shooting times, the quantities of primary interest; the remaining coordinates of $\boldsymbol z$ are state-derivative auxiliaries. 

Let $\boldsymbol z^\star$ be a root of the exact-data system, so that $F(\boldsymbol z^\star,\boldsymbol d)=0$. At $(\boldsymbol z^\star,\boldsymbol d)$, denote the Jacobians by $J_z = \frac{\partial F}{\partial \boldsymbol z}(\boldsymbol z^\star,\boldsymbol d)$ and $J_d = \frac{\partial F}{\partial \boldsymbol d}(\boldsymbol z^\star,\boldsymbol d)$.
Suppose that $J_z$ is nonsingular, and define the sensitivity matrix
\begin{equation}
  S=-P J_z^{-1}J_d.
  \label{eq:fixed-design-sensitivity}
\end{equation}
Thus $S$ is the first order map from errors in the estimated derivatives to errors in the reported model parameters and state values at the shooting times. Recall also that $F$ is polynomial and thus $F \in C^2$.

\begin{theorem}[Local error bound]
\label{thm:fixed-design-local}
There exist a neighborhood $V$ of $\boldsymbol z^\star$ and constants $\rho,C>0$, depending only on $F$, $(\boldsymbol z^\star,\boldsymbol d)$, and $P$.
For $0<\delta<1$, define
\begin{equation}
  r_\delta=\max_{1\le k\le q}\left\{|b_k|
  +\sqrt{2(\Sigma_d)_{kk}\log(2q/\delta)}\right\}.
 \label{eq:fixed-design-radius}
\end{equation}

For every $\delta\in(0,1)$ satisfying $r_\delta<\rho$, the following holds with probability at least $1-\delta$: the perturbed system \[F(\boldsymbol z,\widehat{\boldsymbol d})=0\] has a unique solution $\widehat{\boldsymbol z}$ in $V$, and
\begin{align}
  \|\widehat{\boldsymbol d}-\boldsymbol d\|_\infty
  &\le r_\delta,
  \label{eq:fixed-design-derivative-bound}\\
  \|P(\widehat{\boldsymbol z}-\boldsymbol z^\star)\|_\infty
  &\le \|S\|_\infty r_\delta+C r_\delta^2.
  \label{eq:fixed-design-algebraic-bound}
\end{align}
Here $\rho$ specifies the neighborhood of the exact data in which the algebraic solution is stable, while $r_\delta$ bounds the derivative-estimation error at the chosen confidence level.
\end{theorem}
\begin{proof}
Let $\boldsymbol\xi=\widehat{\boldsymbol d}-\boldsymbol d-\boldsymbol b=W\epsilon$. Then $\boldsymbol\xi\sim\mathcal N(0,\Sigma_d)$. For a coordinate with $(\Sigma_d)_{kk}>0$, set $g=\xi_k/\sqrt{(\Sigma_d)_{kk}}$ and $u=\sqrt{2\log(2q/\delta)}>1$. Since $g\sim\mathcal N(0,1)$, Proposition~2.1.2 of~\cite{Vershynin2026} gives $\Pr(g>u)\le e^{-u^2/2}/\sqrt{2\pi}$. By symmetry and the identity $u^2/2=\log(2q/\delta)$,
\[
\begin{aligned}
 &\Pr\!\left(
 |\xi_k|>\sqrt{2(\Sigma_d)_{kk}\log(2q/\delta)}
 \right)\\
 &\quad=2\Pr(g>u)
 \le \frac{2}{\sqrt{2\pi}}e^{-\log(2q/\delta)}\\
 &\quad=\frac{\delta}{q\sqrt{2\pi}}\le\frac{\delta}{q}.
\end{aligned}
\]
If $(\Sigma_d)_{kk}=0$, then $\xi_k=0$ almost surely and the same bound holds.
Let $A_k$ denote the event that $|\xi_k|$ exceeds the threshold in the preceding display. The union bound gives
\[
 \Pr\!\left(\bigcup_{k=1}^{q} A_k\right)
 \le \sum_{k=1}^{q}\Pr(A_k)
 \le q\cdot\frac{\delta}{q}=\delta.
\]
Consequently, with probability at least $1-\delta$, all $q$ coordinate bounds hold simultaneously.
\[
 \Pr\!\left(
 \begin{gathered}
 |\xi_k|\le\sqrt{2(\Sigma_d)_{kk}\log(2q/\delta)}\\
 \text{for every } k=1,\ldots,q
 \end{gathered}
 \right)\ge 1-\delta.
\]
Adding the deterministic bias gives
\[
 |\widehat d_k-d_k|
 \le |b_k|+\sqrt{2(\Sigma_d)_{kk}\log(2q/\delta)}
 \le r_\delta
\]
for every $k$. This proves~\eqref{eq:fixed-design-derivative-bound}.

For the parameter bound, the implicit function theorem gives neighborhoods $U$ of $\boldsymbol d$ and $V$ of $\boldsymbol z^\star$ and a $C^2$ map $\psi:U\to V$ such that, for every $\boldsymbol d'\in U$, $\psi(\boldsymbol d')$ is the unique solution of $F(\boldsymbol z,\boldsymbol d')=0$ in $V$. Choose $\rho>0$ so that the closed infinity-norm ball of radius $\rho$ about $\boldsymbol d$ lies in $U$. Since
\[
  D(P\psi)(\boldsymbol d)=-P J_z^{-1}J_d=S,
\]
Taylor expansion gives
\[
 P\{\psi(\boldsymbol d+\boldsymbol e)-\psi(\boldsymbol d)\}
 =S\boldsymbol e+R(\boldsymbol e),
 \qquad
 \|R(\boldsymbol e)\|_\infty\le C\|\boldsymbol e\|_\infty^2
\]
for all $\boldsymbol e$ in a sufficiently small ball. On the event in~\eqref{eq:fixed-design-derivative-bound} with $r_\delta<\rho$, taking $\boldsymbol e=\widehat{\boldsymbol d}-\boldsymbol d$ produces the root $\widehat{\boldsymbol z}=\psi(\widehat{\boldsymbol d})$ and the bound~\eqref{eq:fixed-design-algebraic-bound}.
\end{proof}

The corresponding first order term in the Taylor expansion of the local solution map has the joint Gaussian law
\begin{equation}
  S(\widehat{\boldsymbol d}-\boldsymbol d)
  \sim\mathcal N\!\left(
  S\boldsymbol b,\,
  S\Sigma_dS^{\mathsf T}\right).
  \label{eq:fixed-design-linear-law}
\end{equation}
Applying the fixed linear map $S$ to~\eqref{eq:fixed-design-derivative-error} gives this law directly.

The theorem separates derivative estimation error, represented by $(\boldsymbol b,\Sigma_d)$, from the sensitivity of the selected algebraic system, represented by $S$.

%In the synthetic benchmark, the first-order quantities in the theorem are computable: $\boldsymbol b$ can be evaluated from the known noise-free trajectory, $\Sigma_\epsilon$ is the injected noise covariance, and $S$ can be evaluated at the true root. For real data, the bias, the noise covariance, and the sensitivity at the true root are unknown, so we use the theorem as a sensitivity decomposition, not as a calibrated confidence interval.

\subsection{Kernel Noise Gain and Interpretation}

The decomposition also quantifies noise amplification component by component. If the measurement errors for one output are independent with variance $\sigma_\epsilon^2$, then for the component $k$ associated with derivative order $j$ at the point $t$,
\begin{equation}
  (\Sigma_d)_{kk}
  =\sigma_\epsilon^2\|\boldsymbol w_{j,t}\|_2^2,
  \label{eq:fixed-design-noise-gain}
\end{equation}
so $\|\boldsymbol w_{j,t}\|_2$ is the exact gain from the measurement-noise standard deviation to the standard deviation of that derivative estimate. It measures the noise contribution only; the total error also carries the reconstruction bias~\eqref{eq:fixed-design-bias-coordinate}. When $\sigma_{\mathrm{GP}}>0$, the gain admits an a priori bound that depends only on the kernel:
\begin{equation}
  \|\boldsymbol w_{j,t}\|_2\leq
  \frac{\kappa_j}{\sigma_{\mathrm{GP}}},
  \qquad
  \kappa_j^2=
  \left.
  \partial_t^j\partial_{t'}^j k(t,t')
  \right|_{t'=t}.
  \label{eq:fixed-design-kernel-envelope}
\end{equation}
To verify this bound, set $\boldsymbol c=\partial_t^j\boldsymbol k_t$ and $A=\boldsymbol K+\sigma_{\mathrm{GP}}^2\boldsymbol I$. The kernel matrix augmented by this derivative evaluation,
\[
 \begin{pmatrix}
  \kappa_j^2&\boldsymbol c^{\mathsf T}\\
  \boldsymbol c&A
 \end{pmatrix},
\]
is positive semidefinite, so its Schur complement gives $\boldsymbol c^{\mathsf T}A^{-1}\boldsymbol c\le\kappa_j^2$. Since $A\succeq\sigma_{\mathrm{GP}}^2\boldsymbol I$,
\[
 \|\boldsymbol w_{j,t}\|_2^2
 =\boldsymbol c^{\mathsf T}A^{-2}\boldsymbol c
 \le\sigma_{\mathrm{GP}}^{-2}
 \boldsymbol c^{\mathsf T}A^{-1}\boldsymbol c
 \le\frac{\kappa_j^2}{\sigma_{\mathrm{GP}}^2},
\]
which proves~\eqref{eq:fixed-design-kernel-envelope}. Direct differentiation of the kernels~\eqref{eq:ker-se}--\eqref{eq:ker-rq} gives
\begin{equation}
 \begin{aligned}
  \kappa_{j,\mathrm{SE}}
  &=\sigma_f\ell^{-j}
    \sqrt{\frac{(2j)!}{2^j j!}},\\
  \kappa_{j,\mathrm{RQ}}
  &=\sigma_f\ell^{-j}
    \sqrt{\frac{(2j)!(\alpha_{\mathrm{RQ}})_j}
    {j!(2\alpha_{\mathrm{RQ}})^j}},
 \end{aligned}
 \label{eq:fixed-design-kernel-constants}
\end{equation}
where $(\alpha)_j=\alpha(\alpha+1)\cdots(\alpha+j-1)$ is the rising factorial. For these stationary kernels, $\kappa_j$ does not depend on $t$, and the RQ constant reduces to the SE constant as $\alpha_{\mathrm{RQ}}\to\infty$. For fixed $\ell$ and $\sigma_f$, the constants $\kappa_j$ grow rapidly with the derivative order: for the SE kernel, asymptotically as $(2j/e)^{j/2}\ell^{-j}$, up to a constant factor, by Stirling's formula.

\begin{remark}[Practical considerations]
\label{rem:practical}
From the theorem, we make the following practical observations.
\begin{itemize}
    \item Because $\kappa_j$ grows rapidly with $j$, square subsystems that require higher-order output derivatives admit larger worst-case noise gains in~\eqref{eq:fixed-design-kernel-envelope}. This motivates the first criterion of the selection rule in Step~\hyperref[method:step5]{5}: among admissible square subsystems, minimize the highest required derivative order. The realized gain $\|\boldsymbol w_{j,t}\|_2$ also depends on the observation grid, the shooting points, and the fitted matrix $(\boldsymbol K+\sigma_{\mathrm{GP}}^2\boldsymbol I)^{-1}$, and the total error includes the reconstruction bias, so the criterion is a design heuristic rather than a guarantee that the error grows monotonically with the order.
    %\item The second criterion, minimizing the mixed volume, does not appear in the bound at all. It controls the number of homotopy paths and hence the cost of the polynomial solve, not its statistical accuracy.
    %\item Holding the other kernel settings fixed, a larger length scale $\ell$ decreases every positive-order $\kappa_j$ but can increase the reconstruction bias $\boldsymbol b$, the usual bias--variance trade-off. In our pipeline, $\ell$, $\sigma_f$, and $\sigma_{\mathrm{GP}}$ are learned by maximizing the marginal likelihood rather than chosen by hand.
    \item Different square subsystems yield different Jacobians $J_z$ and $J_d$, and hence different sensitivity matrices $S$; at first order, a subsystem with smaller $\|S\|_\infty$ propagates the same derivative error into a smaller parameter error. Our selection rule does not evaluate $S$, since $S$ depends on the unknown solution $\boldsymbol z^\star$. %Note also that $S$ is unchanged if the polynomial equations are multiplied by a nonsingular constant matrix, whereas a condition number of $J_z$ alone depends on such equation scaling; the norm of $S$ still depends on the units of the data and of the algebraic unknowns, so the propagated covariance $S\Sigma_dS^{\mathsf T}$ is more informative than comparing raw Jacobian condition numbers across systems.
\end{itemize}
\end{remark}

In the noise-free setting, interpolation can provide accurate derivatives when the signal is sufficiently smooth, consistent with the comparable performance of AAA-only and the full suite at $\eta=0$. With noisy data, AAA~\cite{NakatsukasaSeteTrefethen2018} interpolates its adaptively selected support values exactly---noise and all---and fits the remaining samples closely, so sample-scale fluctuations can enter the differentiated approximation. GPR instead permits a controlled departure from the observations; for one fixed fit, \eqref{eq:fixed-design-noise-gain} quantifies the resulting measurement-noise gain. The experiment in~\Cref{ssec:aaa_gpr} compares the full suite of derivative estimators with AAA alone, rather than isolating GPR.

\begin{remark}[Scope]
\label{rem:scope}
\Cref{thm:fixed-design-local} describes one set of GPR fits and one shooting configuration, all fixed in advance, and the continuation of the selected root. The pipeline of \Cref{sec:methodology} is adaptive around this core: it learns the GP hyperparameters from the same observations (Step~\hyperref[method:step3]{3}), aggregates candidates across estimators and shooting configurations (Step~\hyperref[step:multi-time-points]{6}), filters and ranks them (Step~\hyperref[method:step7]{7}), and optionally polishes the leading candidates (Step~8). The theorem does not model these adaptive steps, nor the numerical root finding of Step~\hyperref[method:step5]{5} itself, and it bounds the model parameters and state values at the shooting times, not a subsequent conversion of a shooting-time state into an initial condition at $t=0$. We therefore use it to explain the observed behavior, in particular the role of high-order derivatives and of algebraic sensitivity on the hard benchmark systems (\Cref{sec:results}), rather than to certify the complete estimator; calibrated uncertainty quantification for the full pipeline is discussed as future work in \Cref{sec:discussion}.
\end{remark}

\section{Controlled CSTR Example}\label{sec:examples}
We consider a non-isothermal continuous stirred-tank reactor with a known sinusoidal coolant input, following a standard CSTR model~\cite{Seborg} and related examples~\cite{Zhao2021,Kumar2024}. This example is useful because only one output is measured, while the algebraic system must recover both parameters and unobserved state information. The state variables are the scaled reactant concentration $C(t)$, scaled reactor temperature $T(t)$, and an auxiliary effective reaction-rate state $r_{\operatorname{eff}}(t)$. The known input is $u(t)=\sin(0.5t)$, and the measured output is
\[
y_1(t)=700\,T(t).
\]
With coefficients rounded for display, the dynamics are
\begin{align*}
C'(t) &= \frac{1-C(t)}{2\tau}
        - 2\,r_{\operatorname{eff}}(t)C(t),\\
T'(t) &= F_T(C,T,r_{\operatorname{eff}},t),\\
r'_{\operatorname{eff}}(t) &=
        \frac{12.5\,r_{\operatorname{eff}}(t)}{T(t)^2}
        F_T(C,T,r_{\operatorname{eff}},t),
\end{align*}
where
\begin{align*}
F_T(C,T,r_{\operatorname{eff}},t)
&= \frac{T_{\operatorname{in}}-T(t)}{2\tau}
 + 0.02857\,H\,
   r_{\operatorname{eff}}(t)C(t)\\
&\quad -2\,K\,T(t)
 + 0.8571\,K\\
&\quad + 0.05714\,K\sin(0.5t).
\end{align*}
Here $\tau$ is the residence-time parameter, $T_{\operatorname{in}}$ is the scaled inlet temperature, $H=\Delta H/(\rho C_p)$ is the scaled heat-release coefficient, and $K=UA/(V\rho C_p)$ is the scaled heat-transfer coefficient. The auxiliary state $r_{\operatorname{eff}}$ represents the temperature-dependent Arrhenius factor; its differential equation is obtained by differentiating that factor with respect to temperature, which is why the same heat-balance term $F_T$ appears in both $T'$ and $r'_{\operatorname{eff}}$.

The unknown estimated quantities are the four parameters
\[
(\tau,T_{\operatorname{in}},H,K)
\]
and the three initial conditions $(C(0),T(0),r_{\operatorname{eff}}(0))$. Measurements are generated at $750$ uniformly spaced points on $[0,10]$. At a shooting time $t_0$, only the derivatives of the measured output are reconstructed from data: the selected single-point system uses $y_1^{(j)}(t_0)$ for $j=0,\ldots,6$. Since the input $u(t)=\sin(0.5t)$ is known, its derivatives $u^{(j)}(t_0)$ for $j=0,\ldots,6$ are evaluated analytically. These fourteen data quantities enter as coefficients in the selected square polynomial system. The system has $26$ equations in $26$ unknowns, $242$ monomial terms in total ($207$ of them unique), maximum total degree $5$, and mixed volume $134$.

The CSTR example also illustrates a practical failure mode of the estimator. Although the estimated quantities are structurally identifiable, measuring only the temperature makes the problem numerically and statistically difficult. The concentration $C$ and effective reaction rate $r_{\operatorname{eff}}$ are latent, and the heat-balance equation depends on them only through the compound term $H\,C\,r_{\operatorname{eff}}$. Thus a good fit to the temperature trajectory need not imply equally good recovery of each latent factor. %The estimates in \Cref{tab:cstr_failure_mode} are diagnostic of the selected trajectory-fitting solution and of practical identifiability, not constrained physical estimates.

\begin{table}[!t]
\centering
\caption{CSTR latent-factor estimates in one representative benchmark trial. %The compound products $H C_0 r_0$ are computed from the unrounded estimates and may differ from the product of the rounded factors shown. 
The final column is the RMSE between the simulated output from the selected estimate and the noiseless trajectory $y_1^\star(t)$.}
\label{tab:cstr_failure_mode}
\scriptsize
\setlength{\tabcolsep}{3.0pt}
\begin{tabular}{@{}lrrrrr@{}}
\toprule
Noise $\eta$ & $H$ & $C_0$ & $r_0$ & $H C_0 r_0$ & RMSE$(\hat y_1,y_1^\star)$ \\
\midrule
Truth & $0.136$ & $0.647$ & $0.620$ & $0.0546$ & -- \\
$0$ & $0.136$ & $0.647$ & $0.620$ & $0.0546$ & $1.5{\times}10^{-11}$ \\
$10^{-8}$ & $0.136$ & $0.647$ & $0.620$ & $0.0546$ & $2.0{\times}10^{-7}$ \\
$10^{-6}$ & $0.136$ & $0.647$ & $0.620$ & $0.0546$ & $1.8{\times}10^{-5}$ \\
$10^{-4}$ & $0.129$ & $0.636$ & $0.653$ & $0.0534$ & $3.3{\times}10^{-3}$ \\
$10^{-2}$ & $-3.4{\times}10^{-6}$ & $800$ & $-0.267$ & $7.2{\times}10^{-4}$ & $2.4{\times}10^{-1}$ \\
\bottomrule
\end{tabular}
\end{table}

The transition in~\Cref{tab:cstr_failure_mode} is typical of the practical identifiability issue. Through $\eta=10^{-4}$, the output trajectory remains close to the truth and the compound quantity $H C(0)r_{\operatorname{eff}}(0)$ is still close to its true value. At $\eta=10^{-2}$, the fitted-trajectory residual remains small relative to the raw injected noise in this trial, but the individual latent factors are no longer meaningful: the heat-release coefficient is driven close to zero, the concentration initial condition becomes very large, and $r_{\operatorname{eff}}(0)$ changes sign.
This is not a structural-identifiability contradiction; it is a practical-identifiability failure under partial observation and noisy derivative estimation.

The algebraic subsystem is also numerically sensitive. For this CSTR system, the selected $26\times26$ single-point subsystem uses derivatives through order $6$. Its accepted shooting-point instantiations pass numerical rank validation, but the singular-value ratio for the selected polynomial-equation Jacobian is about $9{\times}10^9$ at the validation probes. %This is not a condition number of the final optimizer output; 
This large ratio indicates local numerical sensitivity of the selected subsystem, consistent with the derivative error propagation mechanism in~\Cref{sec:theory}.

\section{Experimental Setup}
\label{sec:experimental_setup}
To validate the performance and robustness of the proposed algebraic method, we conducted a comprehensive benchmark against established parameter estimation methods. This section details the benchmark systems, baseline methods, data generation protocol, and evaluation metrics used in our study.

The accompanying repository includes the benchmark configuration, generated data, reproduction scripts, and combined results CSV used for the paper.

\subsection{Benchmark Systems}
We selected a diverse suite of 25 ODE models from various scientific domains to cover a broad range of systems. The systems were chosen to vary in the number of parameters, number of states, and dynamic behaviors (e.g., oscillations, stiff dynamics). The benchmark suite includes well-known models from ecology (Lotka-Volterra), epidemiology (SEIR), neuroscience (FitzHugh-Nagumo), and immunology (Crauste). Their normalized equations and measured output combinations are summarized in the~\hyperref[sec:benchmark_models_appendix]{Appendix}.

\subsection{Compared Methods}
We compare the proposed method against two established optimization-based estimators and include all three methods here to make the comparison set explicit:
\begin{enumerate}
    \item \textbf{Global Optimization (\texttt{AMIGO2}):} We use AMIGO2~\cite{AMIGO2}, a widely-used toolbox in systems biology that implements a multi-start global optimization approach. This method performs multiple local optimizations from different starting points to better explore the parameter space and avoid local minima.
    \item \textbf{Differential-Evolution Optimization (\texttt{SHADE}):} A success-history based adaptive differential-evolution optimizer~\cite{Tanabe2013}, followed by a local least-squares refinement of the best candidate. Like AMIGO2, it searches for the parameters that minimize the discrepancy between the simulated and observed trajectories.
    \item \textbf{Proposed Method:} The differential-algebraic workflow described in~\Cref{sec:methodology}. Unless otherwise stated, ``proposed method'' denotes the polished variant: algebraic candidates are generated from the configured derivative-estimator suite, ranked by trajectory loss, and then refined by the bounded Levenberg--Marquardt polish step described below.
\end{enumerate}
For AMIGO2 and SHADE, the estimated parameters and initial conditions were searched over the explicit positive box $[10^{-5},10]$; the true values were sampled from $[0.1,0.9]$. The proposed algebraic stage is not a box-constrained search and does not require initial parameter guesses or multi-start strategies.%; candidate filtering and optional polishing used the implementation settings recorded in the archived benchmark artifact.

%We do not include a separate baseline that polishes from random starts without the algebraic stage: AMIGO2 and SHADE already combine global search with local refinement, and thus serve as stronger and more realistic representatives of the optimize-from-random-starts paradigm, against which the proposed method is compared directly.

Beyond these external baselines, two further experiments examine the proposed method itself. The first (\Cref{ssec:aaa_gpr}) compares the proposed no-polish derivative-estimator suite with a variant restricted to AAA rational interpolation only; the latter breaks down rapidly as noise increases. The second (\Cref{ssec:polishing}) isolates the optional local refinement, or ``polishing'', step, whose benefit is substantial and, notably, grows with the noise level.

\subsection{Implementation Details}
For reproducibility, we provide the key implementation details for each method:

\textbf{Proposed method:} The benchmark implementation uses the package's configured derivative-estimator suite, including GPR-based smoothers implemented with \texttt{GaussianProcesses.jl} and \texttt{AbstractGPs.jl}/\texttt{KernelFunctions.jl}, AAA rational interpolation, and Chebyshev approximation variants, with noise-based filtering of interpolation methods known to fail at higher noise. For GPR-based fits, an independent GP is fitted to each observed output component and derivatives are computed via automatic differentiation of the GP posterior mean. The polynomial system solver uses homotopy continuation (HomotopyContinuation.jl). For the polished variant, solutions are refined using bounded Levenberg--Marquardt least squares in per-variable log coordinates, minimizing the same trajectory sum-of-squared-error score used for candidate ranking.

\textbf{AMIGO2:} Global optimization uses the enhanced Scatter Search (eSS) algorithm with a maximum of 200,000 function evaluations and 600-second time limit. Local refinement is performed using the nl2sol algorithm with up to 100,000 iterations and tolerances of $10^{-13}$. ODE integration uses CVODES with absolute and relative tolerances of $10^{-12}$.

\textbf{SHADE:} A success-history based adaptive differential-evolution search is run over the same bounded box with a 600-second time limit and up to 200,000 function evaluations. Its objective is the same trajectory-error loss closure used by the proposed method's polish context; after the global phase, the top five SHADE seeds are passed to the same bounded Levenberg--Marquardt least-squares polish routine used by the proposed polished variant. ODE integration uses the same high-precision settings as the proposed method.

\subsection{Data Generation and Noise Protocol}
For each benchmark system, we generated synthetic data to create a controlled experimental environment where the ground-truth parameters are known.
\begin{enumerate}
    \item \textbf{Parameter Sampling:} For each of 10 experimental trials, true parameter values and initial conditions were sampled uniformly from the interval $[0.1, 0.9]$.
    \item \textbf{Data Simulation:} The ODE system was solved using a high-precision numerical integrator (Vern9 with tolerances $10^{-14}$) to generate a baseline, noise-free trajectory, sampled at 750 equally spaced time points for each system. This dense-sampling setting is appropriate for evaluating derivative-based estimation from time-series data; substantially sparser data are outside the scope of the present benchmark and are discussed as a limitation.
    \item \textbf{Observables:} For each system, only a subset of state variables were treated as observable, reflecting realistic experimental scenarios where not all system states can be directly measured. The measured state variables and output combinations are summarized in the~\hyperref[sec:benchmark_models_appendix]{Appendix}.
    \item \textbf{Noise Injection:} To construct a controlled noisy-data benchmark, we added Gaussian white noise to the noise-free trajectory. We considered one noise-free condition and four nonzero noise levels. For each observed component, independent zero-mean Gaussian noise with standard deviation $\sigma_j=\eta\,|\bar y_j|$ was added at each time point, where $\bar y_j$ is the sample mean of that component's noise-free trajectory and $\eta \in \{0, 10^{-8}, 10^{-6}, 10^{-4}, 10^{-2}\}$. The same value of $\eta$ was applied to all outputs of a given system. Because this scale uses the trajectory mean, $\eta$ is not a uniform signal-to-noise ratio across systems, particularly for outputs with mean near zero.
\end{enumerate}
This protocol results in 25 systems $\times$ 5 noise conditions $\times$ 10 trials = 1{,}250 total datasets per method, allowing us to assess the robustness and average performance of each estimation method.

\subsection{Evaluation Metrics}
To quantify the performance of each method, we use two primary metrics. For these benchmark metrics, the estimated quantities are the structurally identifiable entries among the unknown model parameters and unknown initial conditions. Structurally unidentifiable parameters and initial conditions are excluded a priori from all error calculations.
\begin{enumerate}
    \item \textbf{Success Rate:} A parameter estimation run is considered a ``success'' based on multiple relative error thresholds. We report the percentage of trials where the relative error for every estimated quantity is less than 1\% (SR-1), 10\% (SR-10), and 50\% (SR-50). This multi-threshold approach captures both the precision and the overall reliability of a method.
    \item \textbf{Run-wise worst error:} For each run, we compute the maximum relative error over all estimated quantities. For an estimated quantity $q$, the relative error is defined as $|q_{\text{true}} - q_{\text{estimated}}| / |q_{\text{true}}|$. For runs that failed to produce any result, we assign a penalty error of $10^6$ to the maximum error for that run. We then summarize these per-run worst errors by reporting their median and 90th percentile (P90) across all runs. We use the median because it is robust to occasional outliers. This run-level aggregation ensures that a method is only deemed successful when \emph{all} estimated quantities in a run meet the quality threshold.
\end{enumerate}

Although trajectory error is used internally to rank candidate solutions and as the objective for the optional polishing step, it is not used as the primary benchmark metric. In partially observed ODE systems, large relative errors in the estimated quantities can still produce low trajectory RMSE over the sampled time window, as the CSTR example in \Cref{tab:cstr_failure_mode} illustrates. We therefore evaluate success by direct accuracy of the estimated quantities, using trajectory fit only as an optimization and candidate-selection criterion.

\section{Results}
\label{sec:results}

\subsection{Overall Performance}

\Cref{fig:sr10-noise-all-arms} and \Cref{tab:success_10_noise} summarizes the success rate across all systems as the noise level increases.

Aggregated over all 25 systems and noise levels, the proposed (polished) method attains the highest success rate at every tolerance: SR-1 $79.6\%$, SR-10 $88.5\%$, and SR-50 $91.6\%$, ahead of AMIGO2 ($75.1\%$, $85.0\%$, $88.5\%$) and SHADE ($70.2\%$, $79.1\%$, $83.1\%$); see \Cref{tab:overall_performance}. It is also the most precise, with the lowest median run-wise worst error ($0.0006\%$) and the tightest tail (P90 of $19\%$, versus $91\%$ for AMIGO2 and $378\%$ for SHADE).

\begin{table}[ht]
\centering
\caption{Overall performance with run-level aggregation. SR-X (Success@X\%) is the fraction of runs where all estimated quantities have relative error $<$X\%. Median/P90 Max Error: for each run, the maximum relative error over estimated quantities is computed; the median and 90th percentile across all runs are reported. Failed runs are assigned $10^6$ penalty.}
\label{tab:overall_performance}
\scriptsize
\setlength{\tabcolsep}{2.2pt}
\begin{tabular}{lccccc}
\toprule
Method & SR-1 & SR-10 & SR-50 & \makecell{Median\\Max (\%)} & \makecell{P90\\Max (\%)} \\
\midrule
Proposed (polished) & 79.6 & 88.5 & 91.6 & 0.00058 & 19.0 \\
AMIGO2 & 75.1 & 85.0 & 88.5 & 0.0012 & 91.2 \\
SHADE & 70.2 & 79.1 & 83.1 & 0.0024 & 378.4 \\
\bottomrule
\end{tabular}
\end{table}

The proposed method retains its lead across the \emph{full tested} noise range (\Cref{fig:sr10-noise-all-arms}): SR-10 declines from $99.6\%$ in the noise-free case to $65.2\%$ at the highest noise ($\eta=10^{-2}$), where it still leads AMIGO2 ($58.8\%$) and SHADE ($56.4\%$). All methods drop markedly between $\eta=10^{-4}$ and $\eta=10^{-2}$. Median run-wise worst error by noise level appears in \Cref{tab:noise_performance}; the corresponding SR-10 rates are reported in \Cref{tab:success_10_noise}.

\begin{table}[ht]
\centering
\caption{Median run-wise worst error (\%) by noise level. For each run, the maximum relative error over estimated quantities is computed; the median across runs at each noise level is reported. Failed runs are assigned $10^6$ penalty.}
\label{tab:noise_performance}
\scriptsize
\setlength{\tabcolsep}{2.4pt}
\begin{tabular}{lccccc}
\toprule
Method & 0 & $10^{-8}$ & $10^{-6}$ & $10^{-4}$ & $10^{-2}$ \\
\midrule
Proposed (polished) & $1.30\times 10^{-10}$ & $3.17\times 10^{-6}$ & $3.15\times 10^{-4}$ & 0.03 & 3.06 \\
AMIGO2 & $1.49\times 10^{-7}$ & $4.98\times 10^{-6}$ & $5.08\times 10^{-4}$ & 0.04 & 5.36 \\
SHADE & $5.46\times 10^{-11}$ & $5.36\times 10^{-6}$ & $5.78\times 10^{-4}$ & 0.05 & 5.97 \\
\bottomrule
\end{tabular}
\end{table}

\begin{table}[!t]
\centering
\caption{Success@10\% by noise level: percentage of runs whose maximum relative error over estimated quantities is below 10\%.}
\label{tab:success_10_noise}
\small
\begin{tabular}{lccccc}
\toprule
Method & 0 & $10^{-8}$ & $10^{-6}$ & $10^{-4}$ & $10^{-2}$ \\
\midrule
Proposed (polished) & 99.6 & 98.0 & 93.2 & 86.4 & 65.2 \\
AMIGO2 & 94.8 & 94.8 & 91.2 & 85.2 & 58.8 \\
SHADE & 87.6 & 85.6 & 84.4 & 81.6 & 56.4 \\
\bottomrule
\end{tabular}
\end{table}

\subsection{Performance across systems}

Difficulty varies widely across the 25-system suite. \Cref{tab:high_noise_systems} reports the per-system median run-wise worst error at the highest noise level ($\eta=10^{-2}$), where the differences among methods are most visible. Simple systems (the harmonic oscillator, Van der Pol, and mass-spring-damper) remain accurate for all three headline methods, whereas CSTR, HIV, Crauste, Flexible Arm, and FitzHugh-Nagumo are substantially harder. On these hard systems the proposed (polished) method is often, but not uniformly, the most accurate: it improves over AMIGO2 on Crauste, HIV, and Flexible Arm, while AMIGO2 or SHADE are better on Biohydrogenation, Slow-Fast, and FitzHugh-Nagumo, and all three methods exceed the reporting threshold on CSTR. The optional polishing step is decisive here, lifting several otherwise-unsolved systems from failure to recovery; we examine it in \Cref{ssec:polishing}. Because each run enters through its largest relative error over estimated quantities, these rows should be read as a stringent system-level summary rather than an average over easy and hard quantities.

\begin{table}[!t]
\centering
\caption{Median run-wise worst error (\%) at high noise ($10^{-2}$) by system. For each run, the maximum relative error over estimated quantities is computed; the median across runs is reported.}
\label{tab:high_noise_systems}
\scriptsize
\setlength{\tabcolsep}{2.2pt}
\begin{tabular}{lccc}
\toprule
System & Proposed (polished) & AMIGO2 & SHADE \\
\midrule
Aircraft Pitch & 4.57 & 4.57 & 4.57 \\
Bicycle Model & 0.25 & 0.45 & 0.45 \\
Biohydrogenation & 15.7 & 8.45 & 8.45 \\
Boost Converter & 2.34 & 3.55 & 3.94 \\
Brusselator & 0.12 & 49.2 & 49.2 \\
Crauste & 355.5 & 719.2 & $>1000$ \\
CSTR & $>1000$ & $>1000$ & $>1000$ \\
DAISY MaMil3 & 6.74 & 23.4 & 23.4 \\
DAISY MaMil4 & 1.28 & 2.71 & 2.71 \\
DC Motor & 6.30 & 6.29 & 6.29 \\
FitzHugh-Nagumo & 107.3 & 107.3 & 100.0 \\
Flexible Arm & 112.8 & 289.1 & 230.6 \\
Forced Lotka-Volterra & 0.14 & 0.27 & 0.63 \\
Harmonic Oscillator & 0.005 & 0.005 & 0.005 \\
HIV & 545.0 & $>1000$ & $>1000$ \\
Latent Subpopulation & 0.76 & 9.04 & 9.04 \\
Lotka-Volterra & 0.83 & 0.83 & 0.83 \\
Mass-Spring-Damper & 0.41 & 0.41 & 0.41 \\
Quadrotor & 2.20 & 2.16 & 2.16 \\
Receptor Binding & 6.53 & 12.1 & 12.1 \\
Repressilator & 3.95 & 3.95 & 3.95 \\
SEIR & 2.73 & 5.52 & 5.52 \\
SIRT Treatment & 25.9 & 134.5 & 96.7 \\
Slow-Fast & 3.71 & 2.92 & 2.92 \\
Van der Pol & 0.034 & 0.032 & 0.032 \\
\bottomrule
\end{tabular}
\end{table}

To make the comparison with AMIGO2 more explicit, \Cref{tab:head_to_head_amigo2} breaks down the paired SR-10 outcomes for the proposed (polished) method and AMIGO2 by ODE system.
\begin{table}[!t]
\centering
\caption{Paired SR-10 outcomes for the proposed method and AMIGO2, broken down by ODE system. Each system contributes 50 paired cells (10 sets of parameters and initial conditions at each of five noise levels). ``Proposed only'' counts cells where only the proposed method succeeds; ``AMIGO2 only'' counts the converse, and $\Delta$ is ``Proposed only'' minus ``AMIGO2 only''. Only systems with at least one discordant cell are shown.}
\label{tab:head_to_head_amigo2}
\scriptsize
\setlength{\tabcolsep}{2.6pt}
\begin{tabular}{lrrrrr}
\toprule
System & \makecell{Both\\succeed} & \makecell{Proposed\\only} & \makecell{AMIGO2\\only} & \makecell{Both\\fail} & $\Delta$ \\
\midrule
Crauste & 10 & 17 & 0 & 23 & +17 \\
Brusselator & 27 & 18 & 2 & 3 & +16 \\
Latent Subpopulation & 45 & 5 & 0 & 0 & +5 \\
SIRT Treatment & 39 & 4 & 0 & 7 & +4 \\
Receptor Binding & 43 & 3 & 0 & 4 & +3 \\
DAISY MaMil3 & 44 & 2 & 0 & 4 & +2 \\
CSTR & 23 & 1 & 0 & 26 & +1 \\
Flexible Arm & 37 & 1 & 0 & 12 & +1 \\
Slow-Fast & 48 & 1 & 0 & 1 & +1 \\
Biohydrogenation & 44 & 1 & 2 & 3 & -1 \\
DAISY MaMil4 & 49 & 0 & 1 & 0 & -1 \\
HIV & 18 & 0 & 4 & 28 & -4 \\
\bottomrule
\end{tabular}
\end{table}

This paired view shows that the aggregate advantage is concentrated rather than uniform. Summed over the discordant cells in the table, the net SR-10 advantage is $+44$ paired cells for the proposed method. The largest net gains occur on Crauste and Brusselator, where the proposed method succeeds on many cells for which AMIGO2 does not. At the same time, the comparison is not one-sided: Biohydrogenation and Brusselator each contain individual cells in both directions, and HIV favors AMIGO2 in the paired SR-10 comparison across all noise levels. Thus the aggregate lead should be read as a paired empirical advantage over the benchmark suite, not as a claim that the proposed method dominates AMIGO2 on every ODE instance.

\subsection{Impact of polishing}\label{ssec:polishing}

The proposed method can produce a solution either directly from the algebraic step or after an optional local refinement (``polishing'') of that solution. \Cref{fig:polishing} compares the two variants across noise levels. Polishing is the single largest lever for accuracy: aggregated over the benchmark it raises SR-10 from $70.6\%$ to $88.5\%$. The gain is largest at the highest noise level, where polishing adds $36$ percentage points ($65.2\%$ versus $29.6\%$ at $\eta=10^{-2}$), compared with about $10$ points in the noise-free case. It is especially decisive on the hardest systems, where the raw algebraic solution alone is frequently insufficient.

\begin{figure}[!t]
  \centering
  \includegraphics[width=\columnwidth]{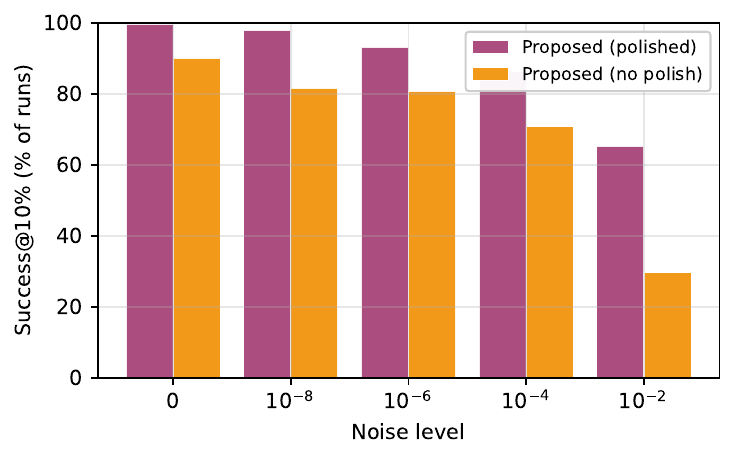}
  \caption{Impact of polishing: success rate (SR-10) by noise level, with and without the optional local refinement. The benefit is largest at the highest noise level.}
  \label{fig:polishing}
\end{figure}

\begin{remark}[Computational performance]
     The implementations of the compared methods differ in optimization level, compilation overhead, and parallelization strategy, so they are not a controlled speed comparison. With that caveat, the optimizers are faster in median time (AMIGO2 $276$\,s, SHADE $167$\,s) than the proposed method ($459$\,s without polishing, $661$\,s with).
\end{remark}

\subsection{Impact of derivative estimation: proposed method vs AAA-only}\label{ssec:aaa_gpr}
\color{black}

To assess the role of derivative estimation before local refinement, we compare the proposed method with polishing disabled against a variant restricted to AAA rational interpolation only, also without polishing. The proposed-method variant uses the benchmark derivative-estimator suite described above, whereas the AAA-only variant uses the derivative estimator from the original algebraic method~\cite{bassik2023robustparameterestimationrational}. Thus the comparison isolates the cost of relying on AAA interpolation alone in noisy data, rather than the effect of local polishing. We evaluate both across the full noise grid (\Cref{fig:aaa_gpr}).

At zero and very small noise levels the two variants are comparable (SR-10 of $91.2\%$ for AAA-only versus $90.0\%$ for the proposed method at $\eta=0$). By $\eta=10^{-4}$, however, the AAA-only variant has degraded sharply: SR-10 falls to $47.2\%$ at $\eta=10^{-4}$ and $10.8\%$ at $\eta=10^{-2}$, against $70.8\%$ and $29.6\%$ for the proposed no-polish method. This shows that restricting derivative estimation to AAA only is far less robust in the presence of measurement noise, and the conclusion holds for the raw algebraic solutions, before any local refinement.

\begin{figure}[!t]
  \centering
  \includegraphics[width=\columnwidth]{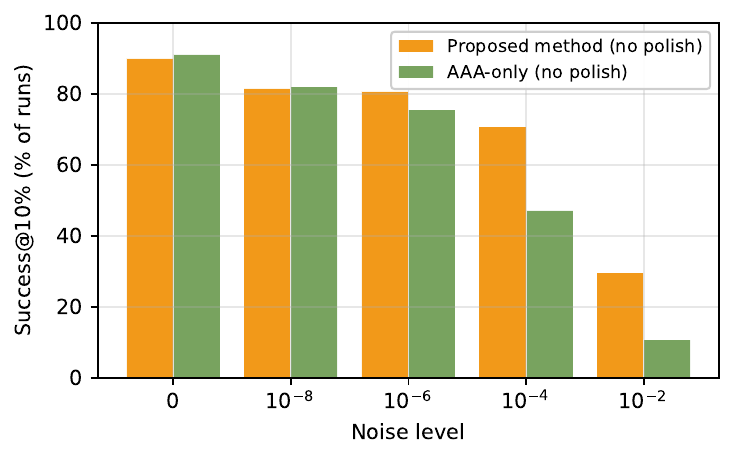}
  \caption{Derivative estimation: proposed method versus AAA-only (both without polishing). The two variants are comparable at zero and very low noise, but AAA-only degrades more sharply at the two highest noise levels.}
  \label{fig:aaa_gpr}
\end{figure}

\section{Discussion}
\label{sec:discussion}
Our results support the viability of algebraic parameter estimation for dense, noisy time-series data when paired with robust derivative estimation. The primary takeaway is that addressing the main practical weakness of the algebraic method (its sensitivity to noisy derivative estimates) makes the approach competitive with established optimization techniques while preserving important advantages of the algebraic formulation.

\subsection{Smoothing and Interpolation under Noise}
The comparison in~\Cref{ssec:aaa_gpr} shows that the full no-polish derivative-estimator suite and the AAA-only variant perform comparably at zero and very low noise, but AAA-only deteriorates much more sharply at the two highest noise levels. This is consistent with the mechanism described in~\Cref{sec:theory}: interpolation can transmit sample-scale noise into the differentiated approximation, whereas regression-based estimators can smooth the observations before differentiation. Because the comparison is between the full suite and AAA alone, it supports the inclusion of smoothing estimators but does not isolate the contribution of GPR.

\subsection{Algebraic candidate generation and local refinement}

The polished method is best viewed as a hybrid procedure: algebraic solving supplies structured candidate parameter sets, and local least-squares refinement turns the best candidates into trajectory-accurate estimates. The same trajectory loss is used to rank and refine the candidates. Because the SHADE baseline applies the same local refinement to its own candidates, the headline comparison reflects the quality of algebraic candidate generation rather than the refinement step alone.

The algebraic stage does not require user-supplied initial parameter guesses or multistart strategies; instead, it supplies data-driven starting points for refinement. Users may still impose physical constraints such as positivity when that information is available.

%Furthermore, because the algebraic stage solves a finite polynomial system, it can enumerate candidate parameter sets for the chosen algebraic subsystem rather than following a single local descent trajectory. With exact derivative data, the algebraic stage can find the complete finite set of parameter solutions predicted by the identifiability analysis~\cite{bassik2023robustparameterestimationrational}. Thus, for a locally but not globally identifiable model, it can return every admissible parameter set rather than whichever one is reached from a single optimizer start.

%This should not be read as a global solution of the noisy trajectory least-squares problem: derivative error, subsystem selection, candidate filtering, and local refinement all affect the final estimate. Instead, the practical role of the algebraic stage is to provide a deterministic and hopefully small set of candidate solutions that can then be ranked and, when desired, refined by the same trajectory loss used for the optimizer baselines.

Performance varies substantially across systems. The harder systems tend to require higher-order output derivatives, larger polynomial systems, or both, as summarized in the appendix. This is consistent with the two mechanisms in~\Cref{sec:theory}: derivative errors are harder to control at higher orders, and the selected algebraic system can amplify those errors.

\subsection{Limitations and Future Work}
The method requires data that are sufficiently dense to support reliable smoothing and differentiation. Our benchmark uses dense synthetic trajectories; sparser measurements could be addressed by incorporating ODE-informed priors or coupling smoothing and parameter estimation more tightly.

The cost of the polynomial solve grows with the number of states, parameters, and equations in the selected template. Extending the method to larger models will require improved equation selection, greater use of model structure, or more specialized polynomial solvers.

The present implementation uses the GPR posterior mean but does not propagate posterior uncertainty through the algebraic solve. The posterior variance could support uncertainty-aware candidate filtering, parameter confidence intervals, or adaptive shooting-point selection. The evaluation is also limited to synthetic data from globally identifiable models with known structure. Real data, model mismatch, and locally identifiable systems are natural directions for future work.

\section{Conclusion}
\label{sec:conclusion}

We have presented a GPR-enhanced differential-algebraic method for estimating parameters and unknown initial conditions in rational ODE systems from noisy time-series data. Across 1{,}250 datasets from 25 systems, the polished method achieved the highest aggregate success rates among the compared methods, with SR-1 of 79.6\%, SR-10 of 88.5\%, and SR-50 of 91.6\%. The controlled comparisons show that the full derivative estimator suite is more robust to noise than AAA alone and that local refinement provides the largest gain on difficult noisy instances. Future work will address sparse and real data, locally identifiable models, parameter uncertainty, and larger polynomial systems.

\paragraph{The use of AI.} We used Claude and Codex to assist in running computational benchmarks, analyzing data, preparing figures, editing, revising the manuscript, and literature search.

\section*{Data and Code Availability}
The code and data used for the benchmark are available at \url{https://github.com/orebas/odepe-noisy-benchmark-artifact/releases/tag/v1.0.1}. The repository includes the analysis scripts, source snapshots, final combined benchmark CSV, configuration files, and a full per-cell benchmark archive containing generated data, run scripts, outputs, logs, metadata, and checksums. Code is distributed under GPL-3.0; generated benchmark data are distributed under CC-BY-4.0.

\appendix
\label{sec:benchmark_models_appendix}
\color{black}
This appendix summarizes the ODE systems used in the benchmark in normalized symbolic form. It records their state and output structure and the selected algebraic template sizes; the exact prescribed inputs, fixed scaling constants, and numerical configurations are provided in the accompanying repository. Unless stated otherwise, the named model coefficients are estimated. The benchmark is adapted from~\cite{bassik2023robustparameterestimationrational} and extended here with systems with prescribed inputs. It comprises 25 systems, including both linear/affine and nonlinear dynamics and spanning several scientific and engineering domains. Each system is globally identifiable for the reported outputs after excluding structurally unidentifiable parameters and initial conditions.

\subsection*{System Specifications and Template Sizes}

Table~\ref{tab:benchmark_template_sizes} summarizes each benchmark system and the algebraic template selected by the proposed method. The reported derivative order and polynomial-system size are selected-template quantities, not intrinsic invariants of the ODE models; the measured state variables and output combinations are listed in the model definitions below.

% Generated by scripts/generate_benchmark_template_table.py.
\begin{table}[!ht]
\centering
\caption{Benchmark systems and selected algebraic template sizes. Outputs is the number of measured quantities, $h$ is the maximum observed-output derivative order, and $|F|$ is the number of equations in the selected square polynomial system, equal to its number of unknowns.}
\label{tab:benchmark_template_sizes}
\scriptsize
\setlength{\tabcolsep}{3.0pt}
\begin{tabular}{@{}lrrrr@{}}
\toprule
System & States & Params & Outputs & $h/|F|$ \\
\midrule
Aircraft Pitch & 3 & 4 & 1 & $5/17$ \\
Bicycle Model & 2 & 3 & 2 & $2/10$ \\
Biohydrogenation & 4 & 6 & 3 & $2/18$ \\
Boost Converter & 2 & 3 & 2 & $2/10$ \\
Brusselator & 2 & 2 & 2 & $1/8$ \\
Crauste & 5 & 12 & 4 & $4/39$ \\
CSTR & 3 & 4 & 1 & $6/26$ \\
DAISY MaMil3 & 3 & 5 & 2 & $4/21$ \\
DAISY MaMil4 & 4 & 7 & 4 & $2/22$ \\
DC Motor & 2 & 2 & 1 & $3/11$ \\
FitzHugh-Nagumo & 2 & 3 & 1 & $4/14$ \\
Flexible Arm & 4 & 5 & 2 & $4/25$ \\
Forced Lotka-Volterra & 2 & 4 & 2 & $2/12$ \\
Harmonic Oscillator & 2 & 2 & 2 & $1/8$ \\
HIV & 5 & 10 & 4 & $3/33$ \\
Latent Subpopulation & 5 & 6 & 5 & $2/23$ \\
Lotka-Volterra & 2 & 3 & 1 & $4/14$ \\
Mass-Spring-Damper & 2 & 3 & 1 & $4/14$ \\
Quadrotor & 2 & 2 & 1 & $3/11$ \\
Receptor Binding & 3 & 6 & 3 & $3/20$ \\
Repressilator & 6 & 3 & 3 & $2/24$ \\
SEIR & 4 & 3 & 3 & $2/17$ \\
SIRT Treatment & 4 & 5 & 3 & $4/26$ \\
Slow-Fast & 6 & 2 & 5 & $2/21$ \\
Van der Pol & 2 & 2 & 2 & $1/8$ \\
\bottomrule
\end{tabular}
\end{table}

\paragraph{Harmonic Oscillator Model}
A model for harmonic oscillators without damping.
\begin{align}
   \begin{cases}
       \dot{x_1} = -a\,x_2 \\
       \dot{x_2} = \frac{1}{b}\,x_1
     \end{cases}
  \label{eq:simple}
\end{align}
\textit{Outputs}: $y_1 = x_1, y_2 = x_2$.

\paragraph{Van der Pol Oscillator Model}
A classical nonlinear oscillator from electrical circuit theory~\cite{VanderPol1927}.
\begin{align}
   \begin{cases}
       \dot{x_1} = a x_2 \\
       \dot{x_2} = -x_1 - b(x_1^2 -1)x_2
     \end{cases}
  \label{eq:vdp}
\end{align}
\textit{Outputs}: $y_1 = x_1, y_2 = x_2$.

\paragraph{FitzHugh-Nagumo Model}
A two-dimensional simplification of spike generation in squid giant axons~\cite{FitzHugh1961,Nagumo1962}.
\begin{align}
   \begin{cases}
       \dot{V} = g\,\left(V - \frac{V^3}{3} + R\right) \\
       \dot{R} = \frac{1}{g}\,\left(V - a + b\,R\right)
     \end{cases}
  \label{eq:fhn}
\end{align}
\textit{Output}: $y_1 = V$.

\paragraph{HIV Dynamics Model}
Models HIV infection dynamics during interaction with the immune system~\cite{Perelson1993}.
\begin{align}
   \begin{cases}
       \dot{x} = \lambda - d\,x - \beta\,x\,v   \\
       \dot{y} = \beta\,x\,v - a\,y             \\
       \dot{v} = k\,y - u\,v                    \\
       \dot{w} = c\,z\,y\,w - c\,q\,y\,w - b\,w \\
       \dot{z} = c\,q\,y\,w - h\,z
     \end{cases}
  \label{eq:hiv}
\end{align}
\textit{Outputs}: $y_1 = w, y_2 = z, y_3 = x, y_4 = y + v$.

\paragraph{Mammillary 3-Compartment Model}
A 3-compartment pharmacokinetic model~\cite{Jacobs1990} from the DAISY identifiability software examples~\cite{Bellu2007}.
\begin{align}{}
   \begin{cases}
       \dot{x_1} = -(a_{21} + a_{31} + a_{01})\,x_1 + a_{12}\,x_2 + a_{13}\,x_3 \\
       \dot{x_2} = a_{21}\,x_1 - a_{12}\,x_2                                    \\
       \dot{x_3} = a_{31}\,x_1 - a_{13}\,x_3
     \end{cases}
  \label{eq:daisy_mamil3}
\end{align}
\textit{Outputs}: $y_1 = x_1, y_2 = x_2$.

\paragraph{Lotka-Volterra Model}
Models predator-prey interactions in an ecosystem~\cite{Lotka1925,Volterra1928}.
\begin{align}
   \begin{cases}
       \dot{r} = k_1\,r - k_2\,r\,w, \\
       \dot{w} = k_2\,r\,w - k_3\,w
     \end{cases}
  \label{eq:lv}
\end{align}
\textit{Output}: $y_1=r$.

\paragraph{Crauste Model}
A CD8 T-cell model~\cite{Terry2012} with the memory-cell decay coefficient fixed at zero.
\begin{align}
   \begin{cases}
       \dot{N} = -\mu_N\,N - \delta_{NE}\,N\,P                                     \\
       \dot{E} = \delta_{NE}\,N\,P - \mu_{EE}\,E^2 - \delta_{EL}\,E + \rho_E\,E\,P \\
       \dot{S} = \delta_{EL}\,E - \delta_{LM}\,S - \mu_{LL}\,S^2 - \mu_{LE}\,E\,S  \\
       \dot{M} = \delta_{LM}\,S                                                    \\
       \dot{P} = \rho_P\,P^2 - \mu_P\,P - \mu_{PE}\,E\,P -\mu_{PL}\,S\,P
     \end{cases}
  \label{eq:crauste}
\end{align}
\textit{Outputs}: $y_1 = N, y_2 = E, y_3 = S + M, y_4 = P$.

\paragraph{Biohydrogenation Model}
Models kinetic processes in the biohydrogenation of fatty acids~\cite{Harvatine2004}.
\begin{align}
   \begin{cases}
       \dot{x_4} = - \frac{k_5 \, x_4 }{k_6 + x_4} \\
       \dot{x_5} = \frac{k_5 \, x_4 }{k_6 + x_4} - \frac{k_7 \, x_5}{k_8 + x_5 + x_6} \\
       \dot{x_6} = \frac{k_7 \, x_5 }{k_8 + x_5 + x_6} - \frac{k_9 \, x_6 \, (k_{10} - x_6) }{k_{10}} \\
       \dot{x_7} = \frac{k_9 \, x_6 \, (k_{10} - x_6) }{k_{10}}
     \end{cases}
  \label{eq:bioh}
\end{align}
\textit{Outputs}: $y_1 = x_4, y_2 = x_5, y_3 = x_6$.

\textit{Note}: The additional observable $y_3 = x_6$ is included so that the system is globally, rather than only locally, identifiable. The state variable $x_7$ is structurally unidentifiable, as it does not appear in the output equations nor influence any observed variables.

\paragraph{Mammillary 4-Compartment Model}
A 4-compartment pharmacokinetic model~\cite{Jacobs1990} from the DAISY identifiability software examples~\cite{Bellu2007}.
\begin{align}
   \begin{cases}
       \dot{x_1}  =  -k_{01} \, x_1 + k_{12} \, x_2 + k_{13}\, x_3 +        \\
       \hspace{3.5em}k_{14}x_4 - k_{21}\, x_1 - k_{31}\, x_1 - k_{41}\, x_1 \\
       \dot{x_2}  =  -k_{12} \, x_2 + k_{21} \, x_1                         \\
       \dot{x_3}  =  -k_{13} \, x_3 + k_{31} \, x_1                         \\
       \dot{x_4}  =  -k_{14} \, x_4 + k_{41} \, x_1                         \\
     \end{cases}
  \label{eq:daisy_mamil4}
\end{align}
\textit{Outputs}: $y_1 = x_1, y_2 = x_2, y_3 = x_3+x_4, y_4 = x_3$. The channel-specific observable $y_4$ is included to make the system globally identifiable.

\paragraph{SEIR Model}
Models an epidemic with stages of disease progression~\cite{Kermack1927}.
\begin{align}
   \begin{cases}
    \dot{S} = -b \, S \, I / N\\
    \dot{E} = b \, S \, I / N - \nu \, E\\
    \dot{I} = \nu \, E - a \, I\\
    \dot{N} = 0
     \end{cases}
  \label{eq:seir}
\end{align}
\textit{Outputs}: $y_1 = I, y_2 = N, y_3 = E$. The exposed-compartment observable $y_3$ is included to make the system globally identifiable.

Several of the remaining benchmark systems are driven by a known input $u(t)$.

\paragraph{Brusselator Model}
A classical model of an autocatalytic chemical reaction with oscillatory dynamics~\cite{Prigogine1968}.
\begin{align}
   \begin{cases}
       \dot{X} = 1 - (b+1)\,X + a\,X^2\,Y \\
       \dot{Y} = b\,X - a\,X^2\,Y
     \end{cases}
  \label{eq:brusselator}
\end{align}
\textit{Outputs}: $y_1 = X, y_2 = Y$.

\paragraph{Repressilator Model}
A synthetic genetic oscillator of three mutually repressing genes~\cite{Elowitz2000}, with mRNA concentrations $m_i$ and protein concentrations $p_i$.
\begin{align}
   \dot m_i &= -m_i + \frac{8\beta}{1+c_i n p_{i-1}},
   \quad i=1,2,3,\notag\\
   \dot p_i &= -2\alpha\left(p_i-\frac{m_i}{s_i}\right),
   \quad i=1,2,3,
  \label{eq:repressilator}
\end{align}
where $p_0=p_3$, $(c_1,c_2,c_3)=(24,16,8)$, and $(s_1,s_2,s_3)=(8,4,12)$. The estimated parameters are $\alpha$, $\beta$, and $n$.
\textit{Outputs}: $y_1 = 4p_1, y_2 = 2p_2, y_3 = 6p_3$.

\paragraph{Forced Lotka-Volterra Model}
The classical predator-prey model driven by a periodic input $u(t)$.
\begin{align}
   \begin{cases}
       \dot{x} = \alpha\,x - \beta\,x\,y + u(t) \\
       \dot{y} = \delta\,x\,y - \gamma\,y
     \end{cases}
  \label{eq:forced_lv}
\end{align}
\textit{Outputs}: $y_1 = x, y_2 = y$.

\paragraph{Latent Subpopulation Model}
A multi-strain epidemic model with three infected subpopulations $I_j$.
\begin{align}
   \begin{cases}
       \dot{S} = -\textstyle\sum_{j=1}^{3} b_j\,S\,I_j \\
       \dot{I_j} = b_j\,S\,I_j - a_j\,I_j, \quad j = 1,2,3 \\
       \dot{R} = \textstyle\sum_{j=1}^{3} a_j\,I_j
     \end{cases}
  \label{eq:latent_subpop}
\end{align}
\textit{Outputs}: $y_1 = S, y_2 = I_1, y_3 = I_2, y_4 = I_3, y_5 = R$.

\paragraph{Receptor Subtype Binding Model}
Ligand binding to two receptor subtypes with distinct rates, where $L$ is the free ligand and $C_a, C_b$ are the bound complexes.
\begin{align}
   \begin{cases}
       \dot{L} = -k_{\mathrm{on}}^{1} L\,(R_1^{\mathrm{tot}} - C_a) + k_{\mathrm{off}}^{1} C_a \\
                 \qquad -k_{\mathrm{on}}^{2} L\,(R_2^{\mathrm{tot}} - C_b) + k_{\mathrm{off}}^{2} C_b \\
       \dot{C_a} = k_{\mathrm{on}}^{1} L\,(R_1^{\mathrm{tot}} - C_a) - k_{\mathrm{off}}^{1} C_a \\
       \dot{C_b} = k_{\mathrm{on}}^{2} L\,(R_2^{\mathrm{tot}} - C_b) - k_{\mathrm{off}}^{2} C_b
     \end{cases}
  \label{eq:receptor}
\end{align}
\textit{Outputs}: $y_1 = L, y_2 = C_a, y_3 = C_b$.

The following are engineering and control systems.

\paragraph{Mass-Spring-Damper System}
A mechanical oscillator with mass $m$, damping $c$, and stiffness $k$, driven by a force $u(t)$, with position $x$ and velocity $v$.
\begin{align}
   \begin{cases}
       \dot{x} = v \\
       m\,\dot{v} = u(t) - c\,v - k\,x
     \end{cases}
  \label{eq:msd}
\end{align}
\textit{Output}: $y_1 = x$.

\paragraph{DC Motor Model}
An armature-controlled DC motor with angular velocity $\omega$ and armature current $i$, driven by an input voltage $u(t)$; $K_t$ and $J_m$ are the estimated torque constant and inertia.
\begin{align}
   \begin{cases}
       J_m\,\dot{\omega} = K_t\,i - b\,\omega \\
       L\,\dot{i} = u(t) - R\,i - k_e\,\omega
     \end{cases}
  \label{eq:dc_motor}
\end{align}
\textit{Output}: $y_1 = \omega$.

\paragraph{Boost Converter Model}
An averaged model of a DC-DC boost converter with inductor current $i_L$ and capacitor voltage $v_C$, known complement $d(t)$ of the switching duty cycle, input voltage $V_{\mathrm{in}}$, inductance $L$, capacitance $C$, and load resistance $R$. The parameters $L$, $C$, and $R$ are estimated, while $V_{\mathrm{in}}$ and $d(t)$ are prescribed.
\begin{align}
   \begin{cases}
       L\,\dot{i_L} = V_{\mathrm{in}} - d(t)\,v_C \\
       C\,\dot{v_C} = d(t)\,i_L - \dfrac{v_C}{R}
     \end{cases}
  \label{eq:boost}
\end{align}
\textit{Outputs}: $y_1 = v_C, y_2 = i_L$.

\paragraph{Quadrotor (Vertical Dynamics)}
A vertical-axis simplification of quadrotor motion with altitude $z$ and vertical velocity $w$, driven by a thrust input $u(t)$, with mass $m$ and drag $d$.
\begin{align}
   \begin{cases}
       \dot{z} = w \\
       m\,\dot{w} = u(t) - d\,w
     \end{cases}
  \label{eq:quadrotor}
\end{align}
\textit{Output}: $y_1 = z$.

\paragraph{Flexible Arm Model}
A two-inertia model of a flexible robot joint with motor angle $\theta_m$, link angle $\theta_t$, their angular velocities $\omega_m, \omega_t$, and coupling stiffness $k$.
\begin{align}
   \begin{cases}
       \dot{\theta_m} = \omega_m \\
       J_m\,\dot{\omega_m} = u(t) - b_m\,\omega_m - k\,(\theta_m - \theta_t) \\
       \dot{\theta_t} = \omega_t \\
       J_t\,\dot{\omega_t} = -b_t\,\omega_t + k\,(\theta_m - \theta_t)
     \end{cases}
  \label{eq:flexible_arm}
\end{align}
\textit{Outputs}: $y_1 = \theta_m, y_2 = \theta_t$.

\paragraph{Aircraft Pitch Model}
A short-period aircraft model with pitch angle $\theta$, pitch rate $q$, angle of attack $\alpha$, and known elevator input $u(t)$.
\begin{align}
   \begin{cases}
       \dot{\theta} = q \\
       \dot{q} = -M_\alpha\,\alpha - M_q\,q - M_{\delta_e}\,u(t) \\
       \dot{\alpha} = q - Z_\alpha\,\alpha
     \end{cases}
  \label{eq:aircraft_pitch}
\end{align}
\textit{Output}: $y_1 = q$. The initial pitch angle $\theta(0)$ is structurally unidentifiable and excluded from scoring.

\paragraph{SIRT Treatment Model}
An SIR-type epidemic model augmented with a treated compartment $T$ that modifies transmission, with total population $N$.
\begin{align}
   \begin{cases}
       \dot{S} = -b\,\dfrac{S\,I}{N} - d\,b\,\dfrac{S\,T}{N} \\
       \dot{I} = b\,\dfrac{S\,I}{N} + d\,b\,\dfrac{S\,T}{N} - (a + g)\,I \\
       \dot{T} = g\,I - \nu\,T \\
       \dot{N} = 0
     \end{cases}
  \label{eq:sirt}
\end{align}
\textit{Outputs}: $y_1 = T, y_2 = N, y_3 = I$.

\paragraph{CSTR Model}
The input-driven continuous stirred-tank reactor model used in the benchmark, including the measured output $y_1=700\,T$, is described in detail in~\Cref{sec:examples}.

\paragraph{Slow-Fast Model}
A slow-fast enzyme-kinetics cascade in which a substrate is converted through intermediates ($x_A \to x_B \to x_C$) on two separated time scales, coupled to enzyme concentrations $e_A, e_B, e_C$ that are constant over the observation window.
\begin{align}
   \begin{cases}
       \dot{x_A} = -k_1\,x_A \\
       \dot{x_B} = k_1\,x_A - k_2\,x_B \\
       \dot{x_C} = k_2\,x_B \\
       \dot{e_A} = \dot{e_B} = \dot{e_C} = 0
     \end{cases}
  \label{eq:slow_fast}
\end{align}
\textit{Outputs}: $y_1 = x_C$, $y_2 = x_A e_A + x_B e_B + x_C e_C$, $y_3 = e_A$, $y_4 = e_C$, $y_5 = e_B$. The observable $y_5 = e_B$ is included to make the system globally identifiable.

\paragraph{Bicycle Model}
A single-track (bicycle) model of vehicle lateral dynamics with lateral velocity $v_y$ and yaw rate $r$, driven by a steering input $u(t)$. Here $C_f, C_r$ are the front and rear cornering stiffnesses and $m$ is the vehicle mass; the forward speed $V_0$, axle distances $a, b$, and yaw inertia $I_z$ are fixed constants.
\begin{align}
   \begin{cases}
       \dot{v_y} = \dfrac{C_f\,\alpha_f + C_r\,\alpha_r}{m} - V_0\,r \\
       \dot{r} = \dfrac{a\,C_f\,\alpha_f - b\,C_r\,\alpha_r}{I_z}
     \end{cases}
  \label{eq:bicycle}
\end{align}
where the tire slip angles are $\alpha_f = u(t) - (v_y + a\,r)/V_0$ and $\alpha_r = -(v_y - b\,r)/V_0$.
\textit{Outputs}: $y_1 = r$, $y_2 = v_y$.

\bibliographystyle{unsrt}
\bibliography{references}

\end{document}